\documentclass[11p,reqno]{amsart}
\usepackage[colorlinks=true, pdfstartview=FitV, linkcolor=blue, citecolor=blue, urlcolor=blue]{hyperref}
\usepackage[usenames,dvipsnames]{xcolor}
\usepackage{latexsym}
\usepackage{amssymb}
\usepackage{amscd}
\usepackage{bm}
\usepackage{amsmath,amsthm}
\usepackage{setspace}
\usepackage{amssymb,amsthm,amsmath,mathrsfs,bm}
\usepackage{graphicx}
\usepackage{subfigure}
\usepackage{appendix}
\usepackage{enumerate}
\usepackage{cite}

\newcommand{\be}{\begin{eqnarray}}
\newcommand{\ee}{\end{eqnarray}}
\newcommand{\bez}{\begin{eqnarray*}}
\newcommand{\eez}{\end{eqnarray*}}

\numberwithin{equation}{section}

\numberwithin{equation}{section}
\newtheorem{theorem}{Theorem}[section]
\newtheorem{lemma}[theorem]{Lemma}
\newtheorem{coro}[theorem]{Corollary}
\newtheorem{prop}[theorem]{Proposition}

\theoremstyle{definition}
\newtheorem{define}[theorem]{Definition}
\newtheorem{remark}[theorem]{Remark}

\DeclareMathOperator{\sgn}{sgn}

\DeclareMathOperator{\Pf}{Pf}

\allowdisplaybreaks[4]
\begin{document}

\title[Isospectral flows related to (dual) cubic strings]{On Isospectral flows related to (dual) cubic strings: Novikov, DP peakons and B, C-Toda lattices}


\date{}

\author{Xiang-Ke Chang}
\thanks{Dedicated to Prof.  Jacek Szmigielski on the occasion of his 70th birthday}
\address{SKLMS \& ICMSEC, Academy of Mathematics and Systems Science, Chinese Academy of Sciences, P.O.Box 2719, Beijing 100190, PR China; and School of Mathematical Sciences, University of Chinese Academy of Sciences, Beijing 100049, PR China.}
\email{changxk@lsec.cc.ac.cn}

\begin{abstract}

The Degasperis--Procesi (DP) equation can be viewed as an isospectral deformation of the boundary value problem for the so-called cubic string, while the Novikov equation can be formally regarded as linked to the dual cubic string. However, their relationships have not been thoroughly investigated. This paper examines various intrinsic connections between these two systems from different perspectives. We uncover a bijective relationship between the DP and Novikov pure peakon trajectories. In particular, this allows us to derive, not previously known, explicit expressions for the constants of motion in the Novikov peakon dynamical system. We also establish a one-to-one correspondence between the corresponding discrete cubic and dual cubic boundary value problems. Furthermore, we propose a new integrable lattice that features bilinear relations involving both determinants and Pfaffians, demonstrating that it can be connected to both the B-Toda and C-Toda lattices, which correspond to isospectral flows, involving positive powers of the spectral parameter, associated with (dual) cubic strings. 

\end{abstract}


\maketitle
\tableofcontents

\section{Introduction}
\subsection{On ordinary (dual) string, CH peakon dynamical system and Toda lattice}
An ordinary \textit{inhomogeneous string} is a second-order equation
\begin{align}\label{string}
-\phi''(y)=zg(y)\phi(y)
\end{align}
equipped with certain boundary condition, which can be derived by the inhomogeneous string wave equation after separation of variables. In 1950s, Krein developed the corresponding direct and inverse problems for boundary value problems of such strings. In particular, he made influential contributions \cite{krein1950inv,krein1950gene}  that generalize Stieltjes’ theory of analytic continued fractions. The notion of the \textit{dual string} was originally introduced by Kac and Krein in \cite{kac1974spectral} and has been discussed at length in the monograph \cite[Section 6.8]{dym1976gaussian} by Dym and McKean. For example, when $g(y)>0$ is a continuous function, a dual string of the string equation \eqref{string} corresponds to a string equation with the mass density function $\tilde g(\tilde y)=\frac{1}{g(y)}$, that is,
\begin{align}\label{dualstring}
- \tilde \phi_{\tilde y\tilde y}(\tilde y)=z\tilde g(\tilde y)\tilde \phi(\tilde y),
\end{align}
where $\tilde y$ is actually the total mass up to the point $y$.
These two systems are related to each other via the change of variables
\begin{subequations}
\begin{align}
&\frac{d\tilde y}{dy}=g(y)=\frac{1}{\tilde g(\tilde y)}, \label{change}\\
& \tilde \phi(\tilde y)=\phi_y(y), \quad \phi(y)=\frac{1}{z}\tilde \phi_{\tilde y}(\tilde y),
\end{align}
\end{subequations}
which can be easily seen if one converts second-order equations to first-order linear systems. The concept of the dual string can also be formulated if $\tilde y$ is a non-decreasing function, 
in particular if it is a non-decreasing piecewise constant function, in which case $g(y)$ and $\tilde g (\tilde y)$ are discrete positive measures with an interchange of the roles of the mass and the length.

There exists an intimate relationship between the string equation \eqref{string} and a 1+1 dimensional shallow water wave equation \cite{camassa1993integrable}
\begin{equation*}\label{eq:ch}
m_t+(um)_x+mu_x=0, \qquad m=u-u_{xx}+\kappa,
\end{equation*}
nowadays called the Camassa--Holm (CH) equation, where $u(x,t)$ is  the wave height at $x$ and time $t$,  and $\kappa$ is a constant. In fact, it was Camassa and Holm in the early of 1990s who derived this equation 
by executing an asymptotic expansion of the Hamiltonian for Euler’s equations of hydrodynamics and proved its integrability in the sense of  Lax pair and bi-Hamiltonian structure etc. In particular, the spatial part of its Lax pair reads
\begin{align}\label{CH_lax_x}
-\psi''(x)+\frac{1}{4}\psi(x)=zm(x)\psi(x)
\end{align}
where $z$ is the spectral parameter. In other words, a Cauchy problem of the CH equation can be associated with an isospectral problem of \eqref{CH_lax_x}. Later on, it was noticed by Beals, Sattinger and Szmigielski \cite{beals1998acoustic,beals2000multipeakons} that the spectral problem \eqref{CH_lax_x} on the line can be transformed into a formal string problem \eqref{string}  on a finite interval via a Liouville transformation.

The CH equation has attracted much attention over the past three decades. As two of the most significant distinctions from the famous KdV equation, the CH equation with $\kappa=0$ permits 
the existence of the so-called peaked soliton solutions (simply called \textit{peakons}) \cite{camassa1993integrable} of the form
\begin{equation}
u(x,t)=\sum_{k=1}^Nm_k(t)e^{-|x-x_k(t)|}, \label{eq:CHpeakon}
\end{equation}
whose dynamics can be characterized by an ODE system,
and
it can lead to a meaningful model of wave-breaking mechanism, which has been pursued by many researchers  (e.g.  \cite{constantin-escher,whitham1974linear}). In addition, it is noted that the spectral problem associated with the CH equation has an energy-dependent potential making the inverse theory more difficult to exploit \cite{bss-string}.  It has been known for some time that peakon solutions capture  main attributes of solutions of the CH equation. For this reason, peakon solutions have been of significant interest;  see a recent survey \cite{lundmark2022view} for some related aspects.

By observing that the CH multipeakon problem is associated with a discrete string problem, i.e. the case where $m$ and $g$ is a finite sum of weighted Dirac measures, Beals et al. \cite{beals1999multipeakons,beals2000multipeakons} explicitly solve
the CH multipeakon ODE system by use of the inverse spectral method involving Stieltjes' theory of analytic continued fractions \cite{stieltjes1894}. As a consequence, the peakon solutions of the CH equation can be expressed in closed form in terms of Hankel determinants involving moments of discrete measures and their associated orthogonal polynomials (OPs). Besides, due to the connection between the string problem and the Jacobi matrix spectral problem, the isospectral flow of which gives the Toda lattice (see e.g. \cite{moser1975finitely} by Moser for perhaps the most relevant work), one can show that the CH peakon and Toda lattices can be viewed as opposite isospectral flows in a well-defined sense  \cite{beals2001peakons,ragnisco1996peakons}.

In fact, the CH problem corresponds to an isospectral deformation of the inhomogeneous
string of the mass density with Dirichlet--Dirichlet boundary condition \cite{beals1998acoustic,beals2000multipeakons}.  For a detailed discussion on isospectral deformations of the mass density in the inhomogeneous string with general boundary conditions,  please refer to \cite{colville2016isospectral} by Colville, Gomez and Szmigielski. In particular, the Dirichlet--Neumann case has appeared, somewhat unexpectedly, in the
work on the two-component modified CH interlacing multipeakons \cite{chang2016multipeakons}. Also see \cite{beals2001inverse} for the Neumann--Neumann case and the Hunter--Saxton equation. This can be actually interpreted as an isospectral deformation of dual string with Dirichlet--Dirichlet boudary condition.

It is known that the dynamics of the CH pure multipeakons has been well understood in \cite{beals2000multipeakons} from the perspective of the definite string.
Actually, in the CH pure multipeakon case, i.e. that all the $m_j(0)$ have the same sign, there will be no collision and one has a unique solution. However, if some of $m_j(0)$ have opposite signs, collisions will incur and the breakdown of regularity can be observed. The characterization of collisions and continuation beyond collisions for multipeakons are more intricate problems (see e.g. \cite{beals2000multipeakons,bressan2007global_con,bressan2007global_diss,holden2008global_diss,holden2007global_cons,xin2000weak,McKean-breakdown}). Concerning global conservative multipeakon
solutions \cite{beals2000multipeakons,holden2007global_cons}, Eckhardt and Kostenko \cite{eckhardt2014isospectral} provided a very compelling interpretation of the mechanism of collisions from the viewpoint of the indefinite string problem due to Krein and Langer \cite{kerin1979/80}. For recent developments on general inverse spectral theory related to generalized indefinite strings, interested readers might want to consult \cite{eckhardt2016inverse,eckhardt2024trace} etc.

\subsection{On (dual) cubic strings,  peakon dynamical systems and Toda-type lattices} 
In recent years, there has been increasing interest in the study of isospectral deformations related to higher-order strings and CH-type equations; see e.g. \cite{beals2023beam,bertola2009cubic,hone2009explicit,lundmark2005degasperis,lundmark2016inverse,kohlenberg2007inverse}. As a third-order analogue of the ordinary string, the \textit{cubic string} is described by the equation
\begin{equation*}
-\phi{'''}(y)=zg(y)\phi(y)
\end{equation*}
with certain boundary conditions, where $z$ is a spectral parameter. It is evident that the ordinary string leads to a self-adjoint problem, while the cubic string is not self-adjoint. There has been long-standing doubt about the applicability of this type of equation to physical systems. However, it was Lundmark and Szmigielski \cite{lundmark2005degasperis} who proposed the concept of the \textit{cubic string} in their study of the multipeakon problem of the Degasperis--Procesi (DP) equation
\begin{equation*} 
  m_t+m_xu+3mu_x=0,   \qquad m=u-u_{xx},
\end{equation*}
which plays a role in shallow water wave theory similar to that of the CH equation \cite{constantin-lannes:hydrodynamical-CH-DP}. It should be noted that, despite being similar in appearance to the CH equation, the DP equation, along with its peakon sector, possesses a distinctly different underlying integrability structure and dynamical characteristics.

The DP equation 
was originally proposed  by Degasperis and Procesi \cite{degasperis-procesi} in a search for equations  satisfying certain asymptotic integrability conditions. Later on, Degasperis, Holm, and Hone \cite{degasperis2002new} demonstrated that
it possesses a formal Lax pair with the spatial part 
 \begin{align*}
 \psi'(x)-\psi'''(x)=zm(x)\psi(x),
 \end{align*}
and admits the multipeakon solution of the form \eqref{eq:CHpeakon}. To solve the multipeakon problem, Lundmark and Szmigielski \cite{lundmark2003multi,lundmark2005degasperis}  introduced the cubic string on a finite interval, which was obtained from  a spatial spectral problem with particular boundary conditions on the line through a Liouville transformation. In fact, the DP multipeakon problem induces an isospectral deformation of the discrete cubic string with a Dirichlet-like boundary condition. 
Although the spectral problem is non-self-adjoint, it turns out that the cubic string fits into the Gantmacher--Krein theory of oscillatory kernels. By exploiting certain Hermite--Pad\'{e} approximation problems, they solved the inverse spectral problems and derived explicit solution formulae in terms of Cauchy bimoment determinants for the DP peakons.

The investigation on DP peakons later motivated the concept of Cauchy biorthogonal polynomials (CBOPs) \cite{bertola2010cauchy,bertola2009cubic}, and also  inspired the development of the Cauchy two-matrix model \cite{bertola2009cauchy, bertola2014cauchy,bertola2014universality} and multi-level Hermite--Pad\'{e} approximations \cite{lago2019mixed}.  Moreover, the so-called Toda lattice of CKP type (C-Toda lattice)  was derived by considering an isospectral deformation of the CBOPs, as a result of which, the C-Toda lattice and the DP multipeakon ODE system can be regarded as opposite flows in some sense \cite{chang2018degasperis}. 

The concept of the \textit{dual cubic string} was  introduced in the study of the Novikov multipeakon problem. The Novikov equation reads
  \begin{equation*}
     m_t+( m_xu+3mu_x)u=0,   \qquad \tilde m=u-u_{xx},
  \end{equation*}
which was first derived by Novikov \cite{novikov2009generalisations}, and was proven to be Lax integrable by Hone and Wang \cite{hone2008integrable}. See also a recent work \cite{chen2022shallow} for its hydrodynamical relevance.
In \cite{hone2009explicit}, Hone, Lundmark and Szmigielski investigated the multipeakon problem by transforming the spatial equation of the Lax pair on the line into a spectral problem on a finite interval and using the inverse spectral method. The transformed spectral problem on a finite interval is nothing but the so-called dual cubic string (see the matrix equation \eqref{eq:dualcubic}). Recent studies also reveal the Novikov peakon ODE system together with the dual cubic string is associated with a class of  so-called partial-skew-orthogonal polynomials (PSOPs), Hermite--Pad\'{e} approximation problems with Pfaffian structures, as well as the Bures random matrix ensemble \cite{chang2022hermite,chang2018partial,chang2018application}.  In addition, an isospectral deformation of the PSOPs can yield the Toda lattice of BKP type (B-Toda lattice)  \cite{chang2018partial}, which could be thought as an opposite flow to the Novikov peakon ODE system \cite{chang2018application}.

It is noted that, the dual cubic string with specific boundary condition relevant to the Novikov peakon ODE system was considered in \cite{hone2009explicit} as a dual of the cubic string with Neumann-like boundary condition for the derivative Burgers equation \cite{kohlenberg2007inverse}, rather than the cubic string with Dirichlet-like boundary conditions for the DP equation \cite{lundmark2005degasperis}.  \textbf{Therefore, it is natural to question whether there exists a correspondence between the dual cubic string boundary problem for the Novikov peakon ODE system and the Dirichlet cubic string for the DP peakon ODE system.}  As an affirmative answer, we shall reveal a correspondence between them in Section \ref{sec:cubic}.

Recall that  the DP equation is connected with a negative flow in the Kaup--Kupershmidt (KK) hierarchy \cite{degasperis2002new}, while the Novikov equation is related to a negative flow in the Sawada--Kotera (SK) hierarchy through reciprocal transformations \cite{hone2008integrable}.  Besides, it is well known, from the classical theory of integrable systems, that \textbf{there exists a fifth-order integrable equation called the modified SK--KK (or Fordy--Gibbons--Jimbo--Miwa) equation \cite{fordy1980some,jimbo1983solitons} that can be  transformed to both of the KK and SK equations via suitable B\"acklund transformations. As far as we know, there has been no such development in the theory of CH-type equations}. As will be claimed in Section \ref{sec:dpnv}, we discover a one-to-one correspondence of the peakon sectors between the DP and Novikov equations, which is somewhat unexpected since it is stronger than expected. Moreover, while it seems that there is no one-to-one correspondence between the B-Toda and C-Toda lattices, in Section \ref{sec:bc}, we propose a new integrable lattice and demonstrate that it can be transformed to both of the B-Toda and C-Toda lattices according to proper B\"acklund transformations, which can be regarded as a discrete analog of the relationship among  the SK and KK as well as the modified SK--KK equations.

The dynamics of pure multipeakons for the DP or Novikov equations have been fully understood in the works \cite{hone2009explicit,lundmark2005degasperis} through the exploration of the theories of definite cubic or dual cubic strings involving  Gantmacher and  Krein's theories \cite{gantmacher-krein} of totally non-negative or even oscillatory kernels.  However, there is still much to be developed regarding indefinite (dual) cubic strings and the mixed peakon-antipeakon problems, despite some existing works, particularly in the peakon-antipeakon sectors e.g. \cite{lundmark2007formation,lundmark2022view,kardell:2016:phdthesis,kardell-lundmark,szmigielski-zhou:shocks-DP, szmigielski-zhou:DP-peakon-antipeakon}. Our results may be generalized for the use of analyzing mixed peakon-antipeakon problems.


\subsection{Contributions and outline of the paper} In summary, our main contributions include:
\begin{enumerate}[(i)]
\item A one-to-one correspondence is revealed between the  discrete cubic string with Dirichlet-like boundary associated with the DP pure multipeakons and the discrete dual cubic string boundary problem associated with the Novikov pure multipeakons.
\item  An unexpected one-to-one correspondence is established between the pure multipeakon trajectories of  the DP and Novikov equations. This enables us to derive, previously unknown, explicit expressions for the constants of motion of the Novikov peakon ODE system.
\item A new integrable lattice is proposed and it is shown to connect with both of the B-Toda and C-Toda lattices via suitable B\"acklund transformations. One of the main key steps is our exploration of three bilinear relations that involve both determinants and Pfaffians.
\end{enumerate}

The paper is organized as follows. We present some preliminary definitions and properties for certain bimoment determinants and Pfaffians in Section \ref{sec:det_pf}. In Section \ref{sec:cubic}, we reveal a one-to-one correspondence between the cubic and dual cubic strings with certain boundary conditions.  In Section \ref{sec:dpnv}, we establish a bijection for the DP and Novikov peakon dynamical systems and some applications are considered.  Section \ref{sec:bc}  is devoted to deriving a new integrable lattice together with B\"acklund transformations to B-Toda and C-Toda lattices.

\section{On certain bimoment determinants and Pfaffians}\label{sec:det_pf}
In this section, we mainly collect some formulae appeared in \cite{bertola2010cauchy,chang2018degasperis,chang2018application,chang2022hermite,hone2009explicit,lundmark2005degasperis,lundmark2016inverse} etc. that are useful for our purpose. In particular, we adopt the same notations used in \cite{chang2022hermite}.
Let's start from the finite moments 
\begin{equation*}
\beta_j=\int x^jd\mu(x)
\end{equation*}
and finite bimoments with respect to the Cauchy kernel $\frac{1}{x + y}$
\begin{equation*}
I_{i,j}=\iint \frac{x^iy^j}{x+y}d\mu(x)d\mu(y),
\end{equation*}
where $d\mu$  be a positive Stieltjes measure on $\mathbb{R}_+$. We also consider
 the 
bimoments based on the skew kernel $\frac{y-x}{x + y}$
\begin{equation*}
J_{i,j}=\iint \frac{y-x}{x+y}x^iy^jd\mu(x)d\mu(y).
\end{equation*}
It is noted that the bimoments $J_{i,j}$ are finite as well once the related bimoments $I_{i,j}$ are finite. In fact, we immediately have the following Lemma.
\begin{prop}There  hold  the following identities for the moments
\begin{align*}
I_{i+1,j}+I_{i,j+1}=\beta_i\beta_j,\qquad J_{i,j}=I_{i,j+1}-I_{i+1,j}=\beta_i\beta_j-2I_{i+1,j}=2I_{i,j+1}-\beta_i\beta_j.
\end{align*}
\end{prop}
\subsection{Cauchy-type bimoment determinants}
Now we introduce certain determinants with the entries involving the bimoments with respect to the Cauchy kernel.
\begin{define}\label{def:FG}
Let $F_k^{(i,j)}$ denote the determinant of the $k\times k$ bimoment matrix which starts with $I_{i,j}$ at the upper left corner, that is,
\begin{align*}
F_k^{(i,j)}=\det(I_{i+p,j+q})_{p,q=0}^{k-1}
=F_k^{(j,i)},
\end{align*}
with the convention that  $F_0^{(i,j)}=1$ and $F_k^{(i,j)}=0$ for $k<0$.

Let $G_k^{(i,j)}$ denote the $k\times k$ determinant
\begin{align*}
G_k^{(i,j)}=
\det(I_{i+p,j+q}\,\,\beta_{i+p})_{\substack{p=0,\ldots,k-1\\ q=0,\ldots,k-2}},
\end{align*}
with the convention that $G_1^{(i,j)}=\beta_i$ and $G_k^{(i,j)}=0$ for $k<1$.

Let $E_k^{(i,j)}$ denote the $k\times k$ determinant
\begin{align*}
E_k^{(i,j)}=
\det\left(\begin{array}{cc}
I_{i+p,j+q}&\beta_{i+p}\\
\beta_{j+q}&0
\end{array}
\right)_{p,q=0}^{k-2},
\end{align*}
with the convention that  $E_k^{(i,j)}=0$ for $k<2$.

\end{define}
Similar to the proof for Heine's formula on Hankel determinants, one can obtain multiple integral expressions for the determinants $F_k^{(i,j)},G_k^{(i,j)}$, from which the positivity properties of these determinants follow.
\begin{prop}\label{prop:FG}
For $k\geq1$, the determinants $F_k^{(i,j)},G_k^{(i,j)}$ admit the integral representations 
\begin{align*}
&F_k^{(i,j)}=\displaystyle\iint_{\substack{0<x_1<\cdots<x_k\\0<y_1<\cdots<y_k}}\frac{\Delta_{[k]}(\mathbf{x})^2\Delta_{[k]}(\mathbf{y})^2}{\Gamma_{[k],[k]}(\mathbf{x},\mathbf{y})}\prod_{p=1}^k\prod_{q=1}^k(x_p)^i(y_q)^jd\mu(x_p)d\mu(y_q),\\
&G_k^{(i,j)}=\displaystyle\iint_{\substack{0<x_1<\cdots<x_k\\0<y_1<\cdots<y_{k-1}}}\frac{\Delta_{[k]}(\mathbf{x})^2\Delta_{[k-1]}(\mathbf{y})^2}{\Gamma_{[k],[k-1]}(\mathbf{x},\mathbf{y})}\prod_{p=1}^k\prod_{q=1}^{k-1}(x_p)^i(y_q)^jd\mu(x_p)d\mu(y_q),
\end{align*}
where $[k]$ denotes the index set $\{1,2,\ldots,k\}$, and, for two index sets $I, J$, 
\begin{align*}
\Delta_J(\mathbf{x})=\prod_{i<j\in J}(x_j-x_i), 
  \quad \Gamma_{I,J}(\mathbf{x};\mathbf{y})=\prod_{i\in I}\prod_{j\in J}(x_i+y_j),  
  \end{align*}
  along with the convention
\begin{align*}
&\Delta_\emptyset(\mathbf{x})=\Delta_{\{i\}}(\mathbf{x})=\Gamma_{\emptyset,J}(\mathbf{x};\mathbf{y})=\Gamma_{I,\emptyset}(\mathbf{x};\mathbf{y})=1.
\end{align*}
Furthermore, 
\begin{enumerate}[(i)]
\item if $\mu(x)$ has infinitely many points of increase, then $F_k^{(i,j)}$ and $G_k^{(i,j)}$ are always positive for any $k\geq1$.
\item if $\mu(x)$ has only finitely many points of increase, say $K$, or equivalently $d\mu(x)=\sum_{i=1}^Ka_i\delta(x-x_i)dx$ with positive $a_i$, then $F_k^{(i,j)}$ and $G_k^{(i,j)}$ are positive for $1\leq k\leq K$, while vanish for any $k>K$.
\end{enumerate}
\end{prop}

\begin{remark}
The detailed proof of the above results can be found in e.g. \cite{bertola2010cauchy,lundmark2016inverse}. 
Due to the multiple integral representations, $F_k^{(i,j)}$ is closely related to the partition function of the Cauchy two-matrix model \cite{bertola2009cauchy,bertola2014cauchy}. 
\end{remark}
\begin{remark}
Some of these quantities  have  appeared in the previous works under different notations. For instance, the quantity $F_k^{(i,j)}$ coincides with the notation $\Delta_k^{(i,j)}$ in \cite{bertola2010cauchy}, and $D_k^{(i,j)}$ in \cite{hone2009explicit}.   $G_k^{(0,1)}$ and $G_k^{(0,2)}$ are equivalent to the quantities $D_k^{'}$  and $D_k^{''}$ in  \cite{hone2009explicit,lundmark2005degasperis}.
The quantity $E_k^{(i,j)}$ was introduced in \cite{chang2018degasperis} using the notation $\rho_k^{(i,j)}$ to express the time evolution of the determinants $F_k^{(i,j)}$ when a time deformation of the measure is given. It will also be useful in the present paper.
\end{remark}

\subsection{Pfaffians related to Cauchy-type bimoment determinants}
 A Pfaffian of order $N$ is defined, based on 
 a skew-symmetric matrix $A=(a_{ij})_{i,j=1}^{2N}$,
according to the formula
\begin{align*}
\Pf(A)=\sum_{P}(-1)^Pa_{i_1,i_2}a_{i_3,i_4}\cdots a_{i_{2N-1},i_{2N}},
\end{align*}
where the summation represents the sum over all possible combinations of pairs selected from the set $\{1,2,\ldots,2N\}$ such that
\begin{align*}
&i_{2l-1}<i_{2l+1},\qquad i_{2l-1}<i_{2l},
\end{align*}
and the factor $(-1)^P$ takes the value $+1$ or  $-1$ depending on whether the sequence $i_1, i_2, . . . , i_{2N}$ is an even or odd permutation of $1, 2, . . . , 2N.$ 
Typically, it is conventionally considered that the Pfaffian of order $0$ is 1, while the Pfaffian of negative order is 0. It is noted that the notation for Pfaffians we employ here is the same as that in \cite{chang2022hermite}, while the Hirota's (slightly different) notation  has been adopted in our previous works \cite{chang2018partial,chang2018application}.

We now introduce certain Pfaffians, that are of different forms depending on parity of the order, invloving entries $J_{i,j}$ and $\beta_j$. For more details, see e.g. \cite[Appendix A]{chang2022hermite} and references therein.

\begin{define} \label{def:tau}
Let $\tau_k^{(l)}$ denote the Pfaffian

\begin{align*}
\tau_{k}^{(l)}=
\left\{
\begin{array}{ll}
\Pf\left(\left(J_{l+i,l+j}\right)_{i,j=0}^{2m-1}\right),&k=2m,\\
\ \\
\Pf\left(\left(
\begin{array}{ccc}
0&\vline&
\begin{array}{c}
\beta_{l+j}
\end{array}\\
\hline
-\beta_{l+i}
&\vline&J_{l+i,l+j}
\end{array}
\right)_{i,j=0}^{2m-1}\right),&k=2m-1,
\end{array}
\right.
\end{align*}
with the convention that $\tau_0^{(l)}=1$ and $\tau_k^{(l)}=0$ for $k<0$. 
\end{define}
 

The positivity properties of these Pfaffians follow immediately from the integral representation (\ref{integral_tau}) below, which can be proved using de Bruijn's formula (see e.g. \cite{chang2018application,chang2018partial}) or Schur's Pfaffian identities (e.g. \cite{chang2022hermite}). 
\begin{prop}
For $k\geq1$, the Pfaffian $\tau_{k}^{(l)}$ admits the integral representation
\begin{equation}\label{integral_tau}
\tau_k^{(l)}=\displaystyle\idotsint\limits_{{0<x_1<\cdots<x_k}}\frac{\Delta_{[k]}(\mathbf{x})^2}{\Gamma_{[k]}(\mathbf{x})}\prod_{p=1}^k(x_p)^ld\mu(x_p), 
\end{equation}
where $[k]$ denotes the index set $\{1,2,\ldots,k\}$, and, for two index sets $I, J$, 
\begin{align*}
\Delta_J(\mathbf{x})=\prod_{i<j\in J}(x_j-x_i), \quad  \Gamma_{J}(\mathbf{x})=\prod_{i<j\in J}(x_j+x_i), 
  \end{align*}
  along with the convention
\begin{align*}
&\Delta_\emptyset(\mathbf{x})=\Delta_{\{i\}}(\mathbf{x})=\Gamma_{\emptyset}(\mathbf{x})=\Gamma_{\{i\}}(\mathbf{x})=1.
\end{align*}
Furthermore, 
\begin{enumerate}[(i)]
\item if $\mu(x)$ has infinitely many points of increase, then $\tau_{k}^{(l)}$ are always positive for any $k\geq1$.
\item if $\mu(x)$ has only finitely many points of increase, say $K$, or equivalently $d\mu(x)=\sum_{i=1}^Ka_i\delta(x-x_i)dx$ with positive $a_i$, then $\tau_{k}^{(l)}$ are positive for $1\leq k\leq K$, while vanish for any $k>K$.
\end{enumerate}
\end{prop}

\begin{remark}
These Pfaffians are closely associated with the partition function for the Bures random matrix ensemble, which was identified from the Bures metric
that can be used to measure the distance of quantum states in quantum mechanics \cite{forrester2016relating}.
\end{remark}

\begin{remark}
The quantities $\tau_k^{(l)}$ have appeared in some of the references (e.g.\cite{hone2009explicit,lundmark2005degasperis}) in terms of the notations $t_k,u_k,v_k$ etc. In particular, 
we have $\tau_k^{(-1)}=t_k,$ $\tau_k^{(0)}=u_k,$ $\tau_k^{(1)}=v_k$. It is in the references \cite{chang2018application,chang2022hermite} that such Pfaffians were identified and the roles of Pfaffians were emphasized in the study of Novikov multipeakons.
\end{remark}

There exist intimate relationships among the determinants $F_k^{(i,j)},G_k^{(i,j)}$ and the Pfaffians $\tau_k^{(l)}$, some of which have appeared in the references (e.g.\cite{bertola2009cauchy,forrester2016relating,hone2009explicit,kohlenberg2007inverse,lundmark2005degasperis}). In \cite{chang2022hermite},  based on the identities connecting determinants and Pfaffians, a new and straightforward approach was employed to derive some old and new formulae. 

\begin{prop}\label{prop:FGtau}
For any integer $k,l$, there hold: 
\begin{equation}\label{rel:FGtau}
\begin{aligned}
&F_k^{(l+1,l)}=F_k^{(l,l+1)}=\frac{(\tau_k^{(l)})^2}{2^k},\qquad G_k^{(l,l+1)}=\frac{\tau_k^{(l)}\tau_{k-1}^{(l)}}{2^{k-1}},  \qquad G_k^{(l,l+2)}=\frac{\tau_k^{(l)}\tau_{k-1}^{(l+1)}}{2^{k-1}},\\
& F_k^{(l,l)}=\frac{1}{2^k}\left(\tau_k^{(l)}\tau_k^{(l-1)}-\tau_{k-1}^{(l)}\tau_{k+1}^{(l-1)}\right),\quad F_k^{(l+1,l-1)}=\frac{1}{2^k}\left(\tau_k^{(l)}\tau_k^{(l-1)}+\tau_{k-1}^{(l)}\tau_{k+1}^{(l-1)}\right). 
\end{aligned}
\end{equation}
\end{prop}
An immediate corollary is as follows.
\begin{coro} \label{coro:FG}
For any $k,l\in \mathbb{Z}$, there hold 
\begin{align*}
(G_k^{(l,l+1)})^2=2F_k^{(l+1,l)}F_{k-1}^{(l+1,l)},\qquad (G_k^{(l,l+2)})^2=2F_k^{(l+1,l)}F_{k-1}^{(l+2,l+1)}.
\end{align*}
\end{coro}

By making use of the Desnanot--Jacobi identity for determinants, several bilinear relations among $E_k^{(i,j)},F_k^{(i,j)},G_k^{(i,j)}$ have been obtained in \cite[Lemma 2.1]{chang2018degasperis}. The following formula will be utilized in this paper.

\begin{prop} \label{lem:bi}
For any $i,j\in \mathbb{Z}$, $k\in  \mathbb{N}_+$, there hold 
\begin{align}
&F_{k+1}^{(i,j)}F_{k-1}^{(i+1,j+1)}=F_{k}^{(i,j)}F_{k}^{(i+1,j+1)}-F_{k}^{(i,j+1)}F_{k}^{(i+1,j)}.\label{id:bi1}
\end{align}
\end{prop}
Moreover, for our purposes, additional new bilinear relations will be derived in the subsection \ref{subsec:detpf}. 
In particular,  Theorem  \ref{coro_new_bi} involves both of determinants and Pfaffians.

\section{Correspondence between discrete cubic string and dual cubic string}\label{sec:cubic}
In this section, we first present some facts on the discrete cubic string with Dirichlet-like boundary conditions  in \cite{lundmark2005degasperis} and the discrete dual cubic string with certain boundary conditions in \cite{hone2009explicit}. Then we shall reveal a one-to-one correspondence between them, which is actually motivated by those in Section \ref{sec:dpnv}.

\subsection{Discrete cubic string} The cubic string 
\begin{align}\label{eq:cubic}
&-\phi{'''}(y)=zg(y)\phi(y),\qquad y\in(-1,1),
\end{align}
with Dirichlet-like boundary conditions 
\begin{equation}\label{cubic:bound}
\phi(-1)=\phi_y(-1)=0,\quad \phi(1)=0
\end{equation}
was introduced in \cite{lundmark2005degasperis} for solving the DP multipeakon problem. More precisely, the corresponding problem is associated with the discrete cubic string, i.e., when $g$ is  a discrete measure
$$g(y)=\sum_{k=1}^Ng_k\delta(y-y_k)dy,\quad g_k>0$$
with $-1=y_0<y_1<\cdots <y_N<y_{N+1}=1$. 

Without loss of generality, they considered the initial value problem
\begin{align*}
&-\phi{'''}(y)=zg(y)\phi(y),\\
&\phi(-1)=\phi_y(-1)=0,\quad \phi_{yy}(-1)=1
\end{align*}
with the discrete measure $g$. It is obvious that $\phi_{yyy}(y)$ is proportional to the delta function at the points $y_1, \ldots, y_N$, and zero elsewhere.  Besides, it is required that $\phi(y)$ and $\phi_y(y)$ to be continuous, and
$\phi_{yy}(y)$ to be piecewise constant with jump discontinuities at $y_1, \ldots, y_N$. Upon denoting by 
$$l_k=y_{k+1}-y_k$$
and 
$$\Phi=(\phi_1,\phi_2,\phi_3)^\top=(\phi,\phi_{y},\phi_{yy})^\top,$$ one will see that there hold the following transition relations
\begin{equation*}
\Phi(y_{k+1}-;z)=L_k\Phi(y_{k}+;z),\quad \Phi(y_{k}+;z)=G_k(z)\Phi(y_{k}-;z)
\end{equation*}
with
\begin{equation*}
L_k=
\begin{pmatrix}
1&l_{k}&\frac{l_{k}^2}{2}\\
0&1&l_{k}\\
0&0&1
\end{pmatrix},\qquad 
G_k(z)=\begin{pmatrix}
1&0&0\\
0&1&0\\
-zg_k&0&1
\end{pmatrix},
\end{equation*}
which result in 
\begin{equation}\label{cubic:tran}
\Phi(1;z)=
\begin{pmatrix}
\phi_1(1;z)\\
\phi_2(1;z)\\
\phi_3(1;z)
\end{pmatrix}=
L_NG_N(z)L_{N-1}G_{N-1}(z)\cdots L_1G_1(z)L_0
\begin{pmatrix}
0\\
0\\
1
\end{pmatrix}.
\end{equation}

The spectral data $\{b_j,\zeta_j\}_{j=1}^N$ is encoded in the Weyl functions defined by
\begin{equation}\label{cubic:weyl}
W(z)=\frac{\phi_2(1;z)}{\phi_1(1;z)},\qquad Z(z)=\frac{\phi_3(1;z)}{\phi_1(1;z)},
\end{equation}
where $\{\zeta_j\}_{j=1}^N$ satisfying $0<\zeta_1<\cdots<\zeta_N$ constitute the zeros of $\phi_1(1;z)$ and $\{b_j\}_{j=1}^N$ are some positive constants satisfying
\begin{subequations}\label{weyl_dis_cubic}
\begin{align}
\frac{W(z)}{z}=&\frac{1}{z}+\sum_{j=1}^N\frac{b_j}{z-\zeta_j}, \\
 \frac{Z(z)}{z}=&\frac{1}{2z}+\sum_{j=1}^N\frac{b_j}{z-\zeta_j}+\sum_{j=1}^N\sum_{k=1}^N\frac{\zeta_jb_kb_j}{(\zeta_k+\zeta_j)(z-\zeta_j)}.
\end{align}
\end{subequations}
Note that it is quite unexpected for this cubic string that the spectrum are simple and positive, since the problem is not self-adjoint and there is no a priori reason for the spectrum to even be real. 

The above procedure results in the forward spectral map from $\{y_k,g_k\}_{k=1}^N$ to $\{b_j,\zeta_j\}_{j=1}^N$. Conversely, $\{y_k,g_k\}_{k=1}^N$ can be uniquely recovered from the spectral data according to the explicit formulae (see \cite[Th. 4.16]{lundmark2005degasperis})
\begin{equation}\label{exp:gy_uvw}
g_{N+1-k}=\frac{(U_k+V_{k-1})^4}{2W_{k-1}W_k},\qquad y_{N+1-k}=\frac{U_k-V_{k-1}}{U_k+V_{k-1}},
\end{equation}
where 
$$
U_k=\tau_{k}^{(0)},\quad V_k=\tau_{k}^{(1)},\quad W_k=\tau_{k}^{(1)}\tau_{k}^{(0)}-\tau_{k-1}^{(1)}\tau_{k+1}^{(0)}=2^kF_k^{(1,1)},
$$
with $\tau_k^{(j)}$ defined in Definition \ref{def:tau} with the discrete measure $d\mu(x)=\sum_{j=1}^Nb_j\delta(x-\zeta_j)dx$. It is noted that ``$U,V,W$" appear as the original notations in \cite[Th. 4.16]{lundmark2005degasperis}, while it has been shown that  they can be equivalently written in terms of $\tau_k^{(j)}$ (see e.g. \cite[Remark 2.8,Theorem 2.9]{chang2022hermite}, or \cite[Corollary 2.5]{chang2018application}).

\subsection{Discrete dual cubic string} The dual cubic string 
\begin{equation}\label{eq:dualcubic}
D_{\tilde y}\begin{pmatrix}
\tilde\phi_1\\
\tilde\phi_{2}\\
\tilde\phi_3
\end{pmatrix}=
\begin{pmatrix}
0& \tilde g(\tilde y)&0\\
0&0&\tilde g(\tilde y)\\
-z&0&0
\end{pmatrix}\begin{pmatrix}
\tilde\phi_1\\
\tilde\phi_{2}\\
\tilde\phi_3
\end{pmatrix},\qquad \tilde y\in(-1,1),
\end{equation}
with the boundary
\begin{equation}\label{dualcubic:bound}
\tilde\phi_{2}(-1)=\tilde\phi_{3}(-1)=0,\quad \tilde\phi_{3}(1)=0,
\end{equation}
was proposed in \cite{hone2009explicit} to deal with Novikov multipeakon problem. Indeed, what was concerned in \cite{hone2009explicit} is the discrete dual cubic string with
$$\tilde g(\tilde y)=\sum_{k=1}^N\tilde g_k\delta(\tilde y-\tilde y_k)d\tilde y,\quad \tilde g_k>0$$
and $-1=\tilde y_0<\tilde y_1<\cdots <\tilde y_N<\tilde y_{N+1}=1$. The duality between \eqref{eq:dualcubic} and \eqref{eq:cubic} manifests itself as follows. When $g(y)>0$ is a continuous function, the systems \eqref{eq:dualcubic} and \eqref{eq:cubic}, regardless of the boundaries,  can be formally related to each other via the same change of the variables
as \eqref{change}, which can be easily demonstrated if one rewrites the third-order equation  \eqref{eq:cubic} as a first-order system. The discrete case corresponds to the case that $\tilde y$ is a piecewise constant function and the duality leads to an interchange of the roles of $g_k$ and $l_k$.
This obviously gives a duality analogue of ordinary dual string introduced by Kac and Krein. See \cite[Remark 5.2]{hone2009explicit} for more details.

Concerning the Novikov peakon problem, without loss of generality, one might consider the discrete dual cubic string \eqref{eq:dualcubic} with the initial value 
$$(\tilde\phi_1(-1),\tilde\phi_{2}(-1),\tilde\phi_3(-1))=(1,0,0).$$   It follows from \eqref{eq:dualcubic} that $\tilde\phi_3$ is continuous and piecewise linear, while $\tilde\phi_1,\tilde\phi_{2}$ are piecewise constants with jumps at $\tilde y_k$. 
More precisely, one can get  the formula for the transition of $\tilde\Phi=(\tilde\phi_1,\tilde\phi_{2},\tilde\phi_3)^T$ from $\tilde y_{k}-$ to $\tilde y_{k+1}-$, that is,
\begin{equation*}
\tilde\Phi(y_k+;z)=\tilde G_k\tilde\Phi(y_k-;z),\qquad \tilde\Phi(y_{k+1}-;z)=\tilde L_k(z)\tilde\Phi(y_k+;z),
\end{equation*}
where
\begin{equation*}
\tilde G_k=
\begin{pmatrix}
1&\tilde g_{k}&\frac{\tilde g_{k}^2}{2}\\
0&1&\tilde g_{k}\\
0&0&1
\end{pmatrix},\qquad 
\tilde L_k(z)=\begin{pmatrix}
1&0&0\\
0&1&0\\
-z\tilde l_k&0&1
\end{pmatrix},
\end{equation*}
with
$
\tilde l_k=\tilde y_{k+1}-\tilde y_k.
$
Then it follows immediately that
\begin{equation*}
\tilde\Phi(1;z)=\tilde L_N(z)\tilde G_N\cdots \tilde L_1(z)\tilde G_1\tilde L_0(z)
\begin{pmatrix}
1\\
0\\
0
\end{pmatrix}.
\end{equation*}

As shown in \cite{hone2009explicit}, the Weyl functions defined by
\begin{equation}\label{dualcubic:weyl}
\begin{aligned}
&\tilde W(z)=-\frac{\tilde\phi_2(1;z)}{\tilde\phi_3(1;z)},\qquad \tilde Z(z)=-\frac{\tilde\phi_1(1;z)}{\tilde\phi_3(1;z)},
\end{aligned}
\end{equation}
encode the spectra data $\{\zeta_j,b_j\}_{j=1}^N$, where $0=\zeta_0<\zeta_1<\zeta_2<\cdots<\zeta_N$ constitute the zeros of $\tilde\phi_3(1;z)=0$ and $\{b_j\}_{j=1}^N$ are some positive constants. In fact, it was shown that $\tilde W(z)$ and $ \tilde Z(z)$ admit the partial fraction decompositions
\begin{equation}\label{weyl_dis_dualcubic}
\tilde W(z)=\sum_{j=1}^N\frac{b_j}{z-\zeta_j}, \qquad  \tilde Z(z)=\frac{1}{2z}+\sum_{j=1}^N\sum_{k=1}^N\frac{b_kb_j}{(\zeta_k+\zeta_j)(z-\zeta_j)}.
\end{equation}
Thus a forward spectral map is established from $\{\tilde y_k,\tilde g_k\}_{k=1}^N$ to $\{b_j,\zeta_j\}_{j=1}^N$. According to \cite[Theorem 7.2]{hone2009explicit} together with the reformulation in \cite[Theorem 5.10]{chang2022hermite}, the inverse spectral problem is uniquely solvable and it
can be explicitly formulated as
\begin{subequations}\label{exp:til_gy}
\begin{align}
\tilde g_{N-k+1}&=\frac{\tau_{k-1}^{(1)}}{\tau_{k}^{(0)}}-\frac{\tau_{k+1}^{(-1)}}{\tau_{k}^{(0)}}-\frac{\tau_{k-2}^{(1)}}{\tau_{k-1}^{(0)}}+\frac{\tau_{k}^{(-1)}}{\tau_{k-1}^{(0)}}=\frac{2^k\left(F_{k}^{(0,0)}+\frac{1}{2}F_{k-1}^{(1,1)}\right)}{\tau_{k}^{(0)}\tau_{k-1}^{(0)}},\\
\tilde l_{N-k}&=\frac{\left(\tau_{k}^{(0)}\right)^4}{2^{2k}\left(F_{k+1}^{(0,0)}+\frac{1}{2}F_k^{(1,1)}\right)\left(F_{k}^{(0,0)}+\frac{1}{2}F_{k-1}^{(1,1)}\right)},\label{exp::til_l}\\
\tilde y_{N-k+1}&=\frac{F_{k+1}^{(0,0)}-\frac{1}{2}F_k^{(1,1)}}{F_{k+1}^{(0,0)}+\frac{1}{2}F_k^{(1,1)}},
\end{align}
\end{subequations}
where $\tau_k^{(j)}$ denote Pfaffians in Definition \ref{def:tau} with the discrete measure $d\mu(x)=\sum_{j=1}^Nb_j\delta(x-\zeta_j)dx$ and it is noted that $F_k^{(l,l)}$ admit Pfaffian representations indicated in Proposition \ref{prop:FGtau}.

\subsection{One-to-one correspondence} 
This subsection is devoted to presenting a one-to-one correspondence between the discrete cubic string \eqref{eq:cubic} with Dirichlet-like boundary conditions \eqref{cubic:bound} and the discrete dual cubic string problem \eqref{eq:dualcubic} with the boundary \eqref{dualcubic:bound}. The main result is given as follows. 


\begin{theorem}
There exists a one-to-one correspondence between the discrete cubic string \eqref{eq:cubic} with Dirichlet-like boundary conditions \eqref{cubic:bound} and the discrete dual cubic string \eqref{eq:dualcubic} with the boundary \eqref{dualcubic:bound}. More precisely,
\begin{enumerate}
\item[(i)] given a discrete cubic string problem \eqref{eq:cubic} with Dirichlet-like boundary conditions \eqref{cubic:bound}, one could construct a discrete dual cubic string problem \eqref{eq:dualcubic} with the boundary \eqref{dualcubic:bound} according to the formulae 
\begin{subequations}\label{gy_to_tilde}
\begin{align}
&\tilde y_k=1-\frac{16}{8+\sum_{j=1}^kg_j(1+y_j)^4},\\
&\tilde g_k=\frac{1}{4}\left(\frac{1}{1+y_k}-\frac{1}{1+y_{k+1}}\right)\left(8+\sum_{j=1}^kg_j(1+y_j)^4\right),
\end{align}
\end{subequations}
such that the both problems produce the same spectral data;
\item[(ii)] given a discrete dual cubic string problem \eqref{eq:dualcubic} with the boundary \eqref{dualcubic:bound},
one could construct a discrete cubic string problem \eqref{eq:cubic} with Dirichlet-like boundary conditions \eqref{cubic:bound} according to the formulae
\begin{subequations}\label{tilde_to_gy}
\begin{align}
&y_k=\frac{4}{2+\sum_{j=k}^N\tilde g_j(1-\tilde y_j)}-1,\\
&g_k=\frac{1}{16}\left(\frac{1}{1-\tilde y_k}-\frac{1}{1-\tilde y_{k-1}}\right)\left(2+\sum_{j=k}^N\tilde g_j(1-\tilde y_j)\right)^4,
\end{align}
\end{subequations}
such that the both problems produce the same spectral data.
\end{enumerate}
\end{theorem}
\begin{proof}
Recall that there exists a bijection between the string data $\{y_j,g_j\}_{j=1}^N$ and the spectral data $\{b_j,\zeta_j\}_{j=1}^N$ for the discrete cubic string problem \eqref{eq:cubic} with Dirichlet-like boundary conditions \eqref{cubic:bound}, and so does that between the string data $\{\tilde y_j,\tilde g_j\}_{j=1}^N$ and the spectral data $\{\tilde b_j,\tilde \zeta_j\}_{j=1}^N$ for discrete dual cubic string problem \eqref{eq:dualcubic} with the boundary \eqref{dualcubic:bound}. In addition, the corresponding recovering formulae are given by \eqref{exp:gy_uvw} and \eqref{exp:til_gy}, respectively. 
This implies that it suffices to verify the connection formulae \eqref{gy_to_tilde} and \eqref{tilde_to_gy} by virtue of \eqref{exp:gy_uvw} and \eqref{exp:til_gy}. 

Now we present the proof to \eqref{gy_to_tilde}. Based on the formulae \eqref{exp:gy_uvw} and the equalities in \eqref{rel:FGtau} as well as  the identity \eqref{id:bi1}, we have 
\begin{align*}
\sum_{j=1}^kg_j(1+y_j)^4=&\sum_{j=N+1-k}^Ng_{N+1-j}(1+y_{N+1-j})^4\\
=&16\sum_{j=N+1-k}^N\frac{(F_j^{(1,0)})^2}{F_j^{(1,1)}F_{j-1}^{(1,1)}}\\
=&16\sum_{j=N+1-k}^N\left(\frac{F_{j}^{(0,0)}}{F_{j-1}^{(1,1)}}-\frac{F_{j+1}^{(0,0)}}{F_j^{(1,1)}}\right)\\
=&16\frac{F_{N+1-k}^{(0,0)}}{F_{N-k}^{(1,1)}},
\end{align*}
which immediately results in the first equality in  \eqref{gy_to_tilde} with the help of the third formula in \eqref{exp:til_gy}.
By calculating the additional term
\begin{align*}
\frac{1}{1+y_k}-\frac{1}{1+y_{k+1}}=\frac{2^{N-k-1}F_{N-k}^{(1,1)}}{\tau_{N+1-k}\tau_{N-k}},
\end{align*}
one can also confirm the second equality in  \eqref{gy_to_tilde} with the help of the equalities in \eqref{rel:FGtau} and the first formula in \eqref{exp:til_gy}.

Furthermore, since it is not hard to see that the maps \eqref{gy_to_tilde} and \eqref{tilde_to_gy} are invertible to each other, the proof is completed. 
\end{proof}

In addition to connecting the cubic and dual cubic string data, we are also able to establish the bridge between the wave functions for the discrete cubic string problem \eqref{eq:cubic} with Dirichlet-like boundary conditions \eqref{cubic:bound} and those for the discrete dual cubic string problem \eqref{eq:dualcubic} with the boundary \eqref{dualcubic:bound}. By recalling the definitions of the corresponding Weyl functions \eqref{cubic:weyl} and \eqref{dualcubic:weyl}, and comparing the partial fraction decompositions \eqref{weyl_dis_cubic} and \eqref{weyl_dis_dualcubic}, in order to ensure that they share the same spectral data, it is sufficent to require that
\begin{equation}\label{rel:wave}
\tilde\Phi(1;z)=T_N
\Phi(1;z),
\end{equation}
with
\begin{equation*}
T_k=\begin{pmatrix}
\frac{1}{2z}+\frac{1}{\tilde l_k}&-\frac{1}{z}&\frac{1}{z}\\
-1&1&0\\
-z&0&0
\end{pmatrix},
\end{equation*}
where  the second and the third rows are clear, and to see the first row, one additionally needs to observe the equality
\begin{align*}
&\sum_{j=1}^N\sum_{k=1}^N\frac{\zeta_jb_kb_j}{(\zeta_k+\zeta_j)(z-\zeta_j)}=z\sum_{j=1}^N\sum_{k=1}^N\frac{b_kb_j}{(\zeta_k+\zeta_j)(z-\zeta_j)}-\sum_{j=1}^N\sum_{k=1}^N\frac{b_kb_j}{(\zeta_k+\zeta_j)},
\end{align*}
and the formula 
\begin{align*}
&\tilde l_N=\frac{1}{F_1^{(0,0)}+\frac{1}{2}}=\frac{2}{2\sum_{j=1}^N\sum_{k=1}^N\frac{b_kb_j}{(\zeta_k+\zeta_j)}+1}
\end{align*}
appeared in \eqref{exp::til_l}.
Furthermore, if $\Phi$ admit the transition rules 
$$
\Phi(y_{k+1}-)=L_k\Phi(y_{k}+),\qquad \Phi(y_{k}+)=G_k\Phi(y_{k}-),
$$
and $\tilde \Phi$ satisfy
$$
\tilde \Phi(\tilde y_{k+1}-)=\tilde L_k\tilde \Phi(\tilde y_{k}+),\qquad \tilde \Phi(\tilde y_{k}+)=\tilde G_k\tilde \Phi(\tilde y_{k}+),
$$
then it is not hard to obtain  the connection formulae 
\begin{align*}
\tilde\Phi(\tilde y_{k+1}-;z)&=\tilde G_{k+1}^{-1}\tilde L_{k+1}^{-1}\cdots\tilde G_N^{-1}\tilde L_N^{-1}T_NL_NG_N\cdots L_{k+1}G_{k+1}
\Phi(y_{k+1}-),\\
 \tilde\Phi(\tilde y_{k}+)&= \tilde L_k^{-1}\tilde G_{k+1}^{-1}\tilde L_{k+1}^{-1}\cdots\tilde G_N^{-1}\tilde L_N^{-1}T_NL_NG_N\cdots L_{k+1}G_{k+1} L_k
\Phi(y_{k}+).
\end{align*}

\begin{remark}
In \cite{hone2009explicit}, it was shown that there exists a nice duality between the discrete dual cubic string problem \eqref{eq:dualcubic} associated with Novikov multipeakons and a discrete cubic string problem with Neumann-like boundary conditions related to the derivative Burgers equation \cite{kohlenberg2007inverse}, which differs from our findings.  Our conclusion here is somewhat unexpected and mainly inspired by the newly observed bijection between DP and Novikov multipeakon ODEs in the following section.
\end{remark}

\section{Correspondence between DP and Novikov multipeakons}\label{sec:dpnv}
The DP and Novikov peakon dynamical systems can be regarded as negative flows obtained by isospectral deformations 
\begin{equation*}
\dot \zeta_j(t)=0,\qquad \dot b_j(t)=\frac{b_j(t)}{\zeta_j}
\end{equation*}
related to the cubic and dual cubic strings \cite{lundmark2005degasperis,hone2009explicit}, respectively. 
In this section, we establish a bijection between these two peakon dynamical systems, and we also make use of the bijective relation to produce some known and unknown results. We claim that the bijective relationship is somehow unexpected.
In fact, the DP peakon ODE system has been interpreted as 
an isospectral flow on a manifold cut out by determinant identities \cite{chang2018degasperis}, while it is more suitable to understand the Novikov peakon ODE system as that on a manifold cut out by 
Pfaffian identities \cite{chang2018application}, due to the fact that their respective hierarchies are associated with different types of infinite dimensional Lie algebras. 
\subsection{DP multipeakons}
The DP equation 
\begin{align}
m_t+(um)_x+2u_xm=0,\qquad m=u-u_{xx},\label{eq:DP}
\end{align}
 was discovered by Degasperis and Procesi \cite{degasperis-procesi} as a new equation satisfying certain asymptotic
integrability conditions. Later on, it was shown by Degasperis, Holm, and Hone \cite{degasperis2002new} to be indeed integrable  in sense of Lax pair, bi-Hamiltonian structure as well as infinitely many conservation laws.  
In particular, the DP equation admits the Lax pair
 \begin{align*}
 &(\partial_x-\partial_x^3)\psi=zm\psi,\\
 &\psi_t=[z^{-1}(1-\partial_x^2)+u_x-u\partial_x]\psi,
 \end{align*}
which means that the DP equation can be formally derived as the compatibility of $\psi_{xxt}=\psi_{txx}$.

When the multipeakon ansatz 
\begin{align}\label{form:DPpeak}
u(x,t)=\sum_{k=1}^Nm_k(t)e^{-|x-x_k(t)|}
\end{align}
is taken into account, it follows from the second equality in \eqref{eq:DP} that $m$ can be regarded as a discrete measure
$$
m(x,t)=2\sum_{k=1}^Nm_k(t)\delta(x-x_k(t)).
$$
Consequently, the first equation of \eqref{eq:DP} is satisfied in a weak sense if the positions $(x_1,\ldots,x_N)$ and momenta $(m_1, \ldots , m_N)$ of the peakons obey the following system of $2N$ ODEs \cite{degasperis2002new,lundmark2005degasperis}:
\begin{subequations}\label{DP_eq:peakon}
\begin{align}
&\dot x_k=u(x_k)=\sum_{j=1}^Nm_je^{-|x_j-x_k|},\\
& \dot m_k=-2m_k\langle u_x\rangle(x_k)=2\sum_{j=1}^N\sgn(x_k-x_j)m_jm_ke^{-|x_j-x_k|}, 
\end{align}
\end{subequations}
where $\langle f \rangle(a)$ denotes the arithemetic average of left and right limits at the point $a$.
It has been shown in \cite{lundmark2005degasperis} that if the initial data $\{x_j(0),m_j(0)\}_{j=1}^N$ are in the phase space
$$
\mathcal{P}=\{\{x_k,m_k\}_{k=1}^N\ |\ -\infty=x_0<x_1<x_2<\cdots<x_N<x_{N+1}=\infty,\ \ \  m_k>0\},
$$
then $\{x_k(t),m_k(t)\}_{k=1}^N$ will uniquely exist and remain in the space $\mathcal{P}$ for all the time $t\in \bf R$ under the flow \eqref{DP_eq:peakon}. This constitutes a pure N-peakon solution that is unique and global. In fact, the solution can be explicitly expressed by use of the inverse spectral method. The following result is mainly due to Lundmark and Szmigielski in \cite[Th. 2.23]{lundmark2005degasperis}, and  it was later reformulated in \cite[Proof to Theorem 2.1]{chang2018degasperis}.

\begin{theorem}The DP equation admits a global N-peakon solution of the form \eqref{form:DPpeak} with
\begin{subequations}\label{sol:DP_FG}
\begin{align}
x_{N-k+1}=&\ln\frac{2F_k^{(1,0)}}{G_k^{(0,2)}}=\ln\frac{G_k^{(0,2)}}{F_{k-1}^{(2,1)}}=\frac{1}{2}\ln\frac{2F_k^{(1,0)}}{F_{k-1}^{(2,1)}}=\ln\frac{\tau_k^{(0)}}{\tau_{k-1}^{(1)}},\\
 m_{N-k+1}=&\frac{(G_{k}^{(0,2)})^2}{2F_{k}^{(1,1)}F_{k-1}^{(1,1)}}=\frac{F_{k}^{(1,0)}F_{k-1}^{(2,1)}}{F_{k}^{(1,1)}F_{k-1}^{(1,1)}}, 
\end{align}
\end{subequations}
where $F_k^{(i,j)}$ and $G_k^{(i,j)}$ are defined in Definition \ref{def:FG}
with the discrete measure 
$$d\mu(x;t)=\sum_{j=1}^Nb_j(t)\delta(x-\zeta_j)dx,$$ and the constants $\zeta_j$ and $b_j(t)$ satisfying
\begin{equation*}
0<\zeta_1<\zeta_2<\cdots<\zeta_N,\qquad 
\dot b_j(t)=\frac{b_j(t)}{\zeta_j}>0.
\end{equation*}
\end{theorem}
\begin{remark}  The DP multipeakon data $\{x_k,m_k\}_{k=1}^N$  in the space $\mathcal{P}$ and the discrete cubic string data  $\{y_k,g_k\}_{k=1}^N$ satisfying  $-1=y_0<y_1<\cdots <y_N<y_{N+1}=1$ and $g_k>0$ form a bijection according to
$$g_k=8m_k\cosh^4  \frac{x_k}{2}, \ \ \ y_k=\tanh \frac{x_k}{2}.$$
See \cite{lundmark2005degasperis} for more details.
\end{remark}

\subsection{Novikov multipeakons}\label{subsec:NVpeakon}
The Novikov equation
 \begin{equation}\label{eq:NV}
{\tilde m}_{t}+{\tilde m}_{\tilde x}{\tilde u}^2+3{\tilde m}{\tilde u}{\tilde u}_{\tilde x}=0,\qquad {\tilde m}={\tilde u}-{\tilde u}_{\tilde x \tilde x},
\end{equation}
was firstly derived by V. Novikov \cite{novikov2009generalisations} using a perturbative symmetry approach and firstly published in the paper by Hone and Wang \cite{hone2008integrable}, who provide its Lax pair 
 \begin{equation*}
D_{\tilde x} \tilde\Psi=U
\tilde\Psi, \qquad D_{t} \tilde\Psi=V\tilde\Psi 
\end{equation*}
with 
\begin{equation*}
U=\begin{pmatrix}
0&\lambda \tilde m &1\\
0&0&\lambda {\tilde m}\\
1&0&0
\end{pmatrix},\quad
V=\begin{pmatrix}
-{\tilde u}{\tilde u}_{\tilde x}&{\lambda }^{-1}{{\tilde u}_{\tilde x}}-\lambda {\tilde u}^2{\tilde m}& {\tilde u}_{\tilde x}^2\\
{\lambda }^{-1}{{\tilde u}}&-\lambda^{-2} &-{\lambda }^{-1}{{\tilde u}_{\tilde x}}-\lambda {\tilde u}^2\tilde m\\
-{\tilde u}^2&{\lambda }^{-1}{{\tilde u}}& {\tilde u}{\tilde u}_{\tilde x}
\end{pmatrix}.
\end{equation*}
In fact, the compatibility condition 
$$(D_{\tilde x}D_{t}- D_{t}D_{\tilde x} )\tilde\Psi= 0$$ 
gives the zero curvature condition
$$ U_{t}-V_{\tilde x}+[U,V]=0,$$
which is exactly the Novikov equation.

The Novikov equation  \eqref{eq:NV} 
admits the multipeakon solution of the form
\begin{equation}\label{form:NVpeak}
{\tilde u}(\tilde x,\tilde t)=\sum_{k=1}^N \tilde m_k(t)e^{-|\tilde x-{\tilde x}_k(t)|}
\end{equation}
in some weak sense if the positions and momenta satisfy the following ODE system:
\begin{subequations}
\begin{align}
&\dot {\tilde x}_{k}={\tilde u}({\tilde x}_k)^2=\left(\sum_{j=1}^N\tilde m_je^{-|{\tilde x}_j-{\tilde x}_k|}\right)^2,\\
& \dot {\tilde m}_{k}=-\tilde m_k{\tilde u}({\tilde x}_k)\langle {\tilde u}_{\tilde x}\rangle ({\tilde x}_k)=\left(\sum_{j=1}^N\tilde m_je^{-|{\tilde x}_j-{\tilde x}_k|}\right)\left(\sum_{j=1}^N\sgn({\tilde x}_k-{\tilde x}_j)\tilde m_je^{-|{\tilde x}_j-{\tilde x}_k|}\right),
\end{align}\label{NV_eq:peakon} 
\end{subequations}
for $1\leq k \leq N$. As is shown in \cite{hone2009explicit}, the pure multipeakon solution of the Novikov equation also possesses a global and unique existence. More exactly, if initial data $\{{\tilde x}_k(0),\tilde m_k(0)\}_{k=1}^N$ are in the phase space
$$
\mathcal{Q}=\{\{{\tilde x}_k,\tilde m_k\}_{k=1}^N\ |-\infty={\tilde x}_0<\ {\tilde x}_1<{\tilde x}_2<\cdots<{\tilde x}_N<{\tilde x}_{N+1}=\infty,\ \ \  \tilde m_k>0\},
$$
then $\{{\tilde x}_k(t),\tilde m_k(t)\}_{k=1}^N$ will uniquely exist and remain in the interior of $\mathcal{Q}$ for all the time $t\in \bf R$ under the flow \eqref{NV_eq:peakon}. The explicit multipeakon formulae can be constructed  by use of inverse spectral methods. The following result was originally derived in \cite{hone2009explicit}, and was later reformulated in terms Pfaffians in \cite{chang2018application,chang2022hermite}.

\begin{theorem}
The NV equation admits a global N-peakon solution of the form \eqref{form:NVpeak} with
\begin{equation}\label{NV_form_yn}
{\tilde x}_{N-k+1}=\frac{1}{2}\ln\frac{Z_k}{W_{k-1}},\qquad \tilde m_{N-k+1}=\frac{\sqrt{Z_kW_{k-1}}}{U_kU_{k-1}}.
\end{equation}
Here 
$$
U_k=\tau_{k}^{(0)},\quad W_k=\tau_{k}^{(1)}\tau_{k}^{(0)}-\tau_{k-1}^{(1)}\tau_{k+1}^{(0)}=2^kF_k^{(1,1)},\quad Z_k=\tau_{k}^{(-1)}\tau_{k}^{(0)}-\tau_{k+1}^{(-1)}\tau_{k-1}^{(0)}=2^kF_k^{(0,0)},
$$
where $\tau_k^{(j)}$ denote Pfaffians in Definition \ref{def:tau} with the discrete measure 
$$d\mu(x;t)=\sum_{j=1}^Nb_j(t)\delta(x-\zeta_j)dx$$
satisfying 
\begin{equation*}
0<\zeta_1<\zeta_2<\cdots<\zeta_N,\qquad 
b_j(t)=b_j(0)e^{\frac{t}{\zeta_j}}>0.
\end{equation*}

\end{theorem}

\begin{remark}
In the pure peakon sector, one has $\tilde m(\tilde x)=2\sum_{k=1}^N\tilde m_k\delta(\tilde x-\tilde x_k),\, \tilde m_k>0$, and there exists a one-to-one correspondence between the Novikov multipeakon data   in the space
$\mathcal{Q}$ and  the discrete dual cubic string data satisfying $-1=\tilde y_0<\tilde y_1<\cdots <\tilde y_N<\tilde y_{N+1}=1$ and $\tilde g_k>0$. The one-to-one map is given by
\begin{equation*}
 \tilde g_{i}=2\tilde m_k\cosh {\tilde x_k},\quad \tilde y_k=\tanh \tilde x_k.
\end{equation*}
See \cite{hone2009explicit} for more details.
\end{remark}

\subsection{One-to-one correspondence} In this subsection, we establish a bijection between the DP peakon trajectory \eqref{DP_eq:peakon} in the space  $\mathcal{P}$ and the NV peakon trajectory \eqref{NV_eq:peakon} in  $\mathcal{Q}$.

\begin{theorem}\label{th:dpnv}
If a set of time-dependent data $\left\{x_k(t),m_k(t)\right\}_{k=1}^N$ in the space $$
\mathcal{P}=\{\{x_k,m_k\}_{k=1}^N\ |\ -\infty=x_0<x_1<x_2<\cdots<x_N<x_{N+1}=\infty,\ \ \  m_k>0\},
$$ satisfy the DP peakon ODE system \eqref{DP_eq:peakon}, then  $\{{\tilde x}_k(t),\tilde m_k(t)\}_{k=1}^N$ obtained by
\begin{align}\label{xm_to_yn}
{\tilde x}_k=\frac{1}{2}\ln\left(\sum_{j=1}^km_je^{2x_j}\right),\qquad \tilde m_k=\left(e^{-x_k}-e^{-x_{k+1}}\right)e^{{\tilde x}_k}
\end{align}
lie within the space 
$$
\mathcal{Q}=\{\{{\tilde x}_k,\tilde m_k\}_{k=1}^N\ |-\infty={\tilde x}_0<\ {\tilde x}_1<{\tilde x}_2<\cdots<{\tilde x}_N<{\tilde x}_{N+1}=\infty,\ \ \  \tilde m_k>0\},
$$
 and satisfy the Novikov peakon ODE system \eqref{NV_eq:peakon}. 
 
Conversely, given a set of data  $\{{\tilde x}_k(t),\tilde m_k(t)\}_{k=1}^N$ in the space $\mathcal{Q}$ satisfying the Novikov peakon ODE system \eqref{NV_eq:peakon}, one can construct $\left\{x_k(t),m_k(t)\right\}_{k=1} 
^N$  according to
\begin{align}\label{yn_to_xm}
 x_k=-\ln\left(\sum_{j=k}^N\tilde m_je^{-{\tilde x}_j}\right),\qquad m_k=\left(e^{2{\tilde x}_k}-e^{2{\tilde x}_{k-1}}\right)e^{-2 x_k}
\end{align}
so that $\{{ x}_k(t), m_k(t)\}_{k=1}^N$ belong to $\mathcal{P}$ and satisfy the DP peakon ODE system \eqref{DP_eq:peakon}.

Consequently, the DP peakon trajectory \eqref{DP_eq:peakon} in the space  $\mathcal{P}$ and the NV peakon trajectory \eqref{NV_eq:peakon} in  $\mathcal{Q}$ form a bijection.
\end{theorem}
\begin{proof}
It is sufficient to provide a proof in a local time since these two pure peakon flows exist globally.

We first present the proof of the map from the DP peakon trajectory to the NV peakon trajectory. If $\left\{x_k,m_k\right\}_{k=1}^N$ are in $\mathcal{P}$, then it is easy to see from \eqref{xm_to_yn} that $\left\{{\tilde x}_k,\tilde m_k\right\}_{k=1}^N$  lie within $\mathcal{Q}$.  In order to prove that $\left\{{\tilde x}_k(t),\tilde m_k(t)\right\}_{k=1}^N$ satisfy the Novikov peakon ODE system \eqref{NV_eq:peakon}, it is sufficient to demonstrate the equalities
\begin{subequations}\label{pf_nv_yn1}
\begin{align}
&\frac{d}{dt}\left(\frac{1}{2}\ln\left(\sum_{i=1}^km_ie^{2x_i}\right)\right)\nonumber\\ 
=&\left(\frac{\sum_{j=1}^k\left(e^{-x_j}-e^{-x_{j+1}}\right)\sum_{i=1}^jm_ie^{2x_i}}{\sqrt{\sum_{i=1}^km_ie^{2x_i}}}+\sum_{j=k+1}^N\left(e^{-x_j}-e^{-x_{j+1}}\right)\sqrt{\sum_{i=1}^km_ie^{2x_i}}\right)^2,\label{pf_nv_y1}\\
&\frac{d}{dt}\ln\left(e^{-x_k}-e^{-x_{k+1}}\right)+\frac{d}{dt}\left(\frac{1}{2}\ln\left(\sum_{j=1}^km_je^{2x_j}\right)\right)\nonumber \\
=&\left(\frac{\sum_{j=1}^k\left(e^{-x_j}-e^{-x_{j+1}}\right)\sum_{i=1}^jm_ie^{2x_i}}{\sqrt{\sum_{i=1}^km_ie^{2x_i}}}+\sum_{j=k+1}^N\left(e^{-x_j}-e^{-x_{j+1}}\right)\sqrt{\sum_{i=1}^km_ie^{2x_i}}\right)\nonumber\\
&\left(\frac{\sum_{j=1}^{k-1}\left(e^{-x_j}-e^{-x_{j+1}}\right)\sum_{i=1}^jm_ie^{2x_i}}{\sqrt{\sum_{i=1}^km_ie^{2x_i}}}-\sum_{j=k+1}^N\left(e^{-x_j}-e^{-x_{j+1}}\right)\sqrt{\sum_{i=1}^km_ie^{2x_i}}\right)\label{pf_nv_n1},
\end{align}
\end{subequations}
which is derived by substituting \eqref{xm_to_yn} into \eqref{NV_eq:peakon}. Under the assumption that $\left\{x_k(t),m_k(t)\right\}_{k=1}^N$ satisfy the DP peakon ODE system \eqref{DP_eq:peakon}, \eqref{pf_nv_y1} becomes
 \begin{align*}
& \sum_{i=1}^k\left(2m_i\sum_{j=1}^{i-1}m_je^{x_j-x_i}+m_i^2\right)e^{2x_i}\\
=& \left(\sum_{j=1}^k\left(e^{-x_j}-e^{-x_{j+1}}\right)\sum_{i=1}^jm_ie^{2x_i}+\sum_{j=k+1}^N\left(e^{-x_j}-e^{-x_{j+1}}\right)\sum_{i=1}^km_ie^{2x_i}\right)^2.
 \end{align*}
 It is not hard to prove that the above equality holds. In fact, both sides of the equation can be simplified into $\left(\sum_{j=1}^km_je^{x_j}\right)^2$. 
 
 As for \eqref{pf_nv_n1},  by subtracting \eqref{pf_nv_y1} from \eqref{pf_nv_n1}, one can arrive at 
 \begin{align*}
 &\frac{d}{dt}\ln\left(e^{-x_k}-e^{-x_{k+1}}\right)=\frac{-\dot x_ke^{-x_k}+\dot x_{k+1}e^{-x_{k+1}}}{e^{-x_k}-e^{-x_{k+1}}}\nonumber \\
=&-\left(\frac{\sum_{j=1}^k\left(e^{-x_j}-e^{-x_{j+1}}\right)\sum_{i=1}^jm_ie^{2x_i}}{\sqrt{\sum_{i=1}^km_ie^{2x_i}}}+\sum_{j=k+1}^N\left(e^{-x_j}-e^{-x_{j+1}}\right)\sqrt{\sum_{i=1}^km_ie^{2x_i}}\right)\nonumber\\
&\left(\frac{\left(e^{-x_k}-e^{-x_{k+1}}\right)\sum_{i=1}^km_ie^{2x_i}}{{\sqrt{\sum_{i=1}^km_ie^{2x_i}}}}+2\sum_{j=k+1}^N\left(e^{-x_j}-e^{-x_{j+1}}\right)\sqrt{\sum_{i=1}^km_ie^{2x_i}}\right).
 \end{align*}
In the case that $\left\{x_k(t),m_k(t)\right\}_{k=1}^N$ undergo the evolution \eqref{DP_eq:peakon}, the above formula can be simplified into
\begin{align*}
&\frac{-\sum_{j=1}^km_je^{x_j-2x_k}+\sum_{j=1}^km_je^{x_j-2x_{k+1}}}{e^{-x_k}-e^{-x_{k+1}}}\cdot \left(\sum_{i=1}^km_ie^{2x_i}\right)\\
=&-\left(\sum_{j=1}^km_je^{x_j}\right)\left(\left(e^{-x_k}-e^{-x_{k+1}}\right)\sum_{i=1}^km_ie^{2x_i}+2\sum_{j=k+1}^N\left(e^{-x_j}-e^{-x_{j+1}}\right){\sum_{i=1}^km_ie^{2x_i}}\right),
\end{align*}
which is nothing but an identity. Therefore, the map from the  DP to  Novikov peakon trajectories is confirmed.

 Conversely, if
  $\left\{{\tilde x}_k,\tilde m_k\right\}_{k=1}^N$ are in $\mathcal{Q}$, then it easily follows from \eqref{yn_to_xm} that $\left\{x_k,m_k\right\}_{k=1}^N$  lie within $\mathcal{P}$.   What we are left to prove  is that $\left\{x_k(t),m_k(t)\right\}_{k=1}^N$ generated by \eqref{yn_to_xm} undergo the DP peakon trajectory under the assumption of $\left\{{\tilde x}_k(t),\tilde m_k(t)\right\}_{k=1}^N$ satisfying the Novikov peakon ODE system \eqref{NV_eq:peakon}.
In fact, after the substitution \eqref{yn_to_xm}, what we need to prove is 
\begin{subequations}\label{pf_dp_xm1}
\begin{align}
&-\frac{d}{dt}\left(\ln\left(\sum_{j=k}^N\tilde m_je^{-{\tilde x}_j}\right)\right)\nonumber\\
&=\sum_{j=1}^k\left(e^{2{\tilde x}_j}-e^{2{\tilde x}_{j-1}}\right)\left(\sum_{i=j}^N\tilde m_ie^{-{\tilde x}_i}\right)\left(\sum_{i=k}^N\tilde m_ie^{-{\tilde x}_i}\right)\nonumber\\
&\qquad\qquad\,\,+\sum_{j=k+1}^N\left(e^{2{\tilde x}_j}-e^{2{\tilde x}_{j-1}}\right)\left(\sum_{i=j}^N\tilde m_ie^{-{\tilde x}_i}\right)^3\bigg/\left(\sum_{i=k}^N\tilde m_ie^{-{\tilde x}_i}\right),\label{pf_dp_x1}\\
&\frac{d}{dt}\ln\left(e^{2{\tilde x}_k}-e^{2{\tilde x}_{k-1}}\right)+2\frac{d}{dt}\left(\ln\left(\sum_{j=k}^N\tilde m_je^{-{\tilde x}_j}\right)\right)\nonumber \\ =&2\sum_{j=1}^{k-1}\left(e^{2{\tilde x}_j}-e^{2{\tilde x}_{j-1}}\right)\left(\sum_{i=j}^N\tilde m_ie^{-{\tilde x}_i}\right)\left(\sum_{i=k}^N\tilde m_ie^{-{\tilde x}_i}\right)\nonumber\\
&\qquad\qquad\,\,-2\sum_{j=k+1}^N\left(e^{2{\tilde x}_j}-e^{2{\tilde x}_{j-1}}\right)\left(\sum_{i=j}^N\tilde m_ie^{-{\tilde x}_i}\right)^3\bigg/\left(\sum_{i=k}^N\tilde m_ie^{-{\tilde x}_i}\right).\label{pf_dp_m1}
\end{align}
\end{subequations}
Under the evolution \eqref{NV_eq:peakon} for $\left\{{\tilde x}_k(t),\tilde m_k(t)\right\}_{k=1}^N$, the equation \eqref{pf_dp_x1} becomes
\begin{align*}
& \sum_{j=k}^N\left(\tilde m_j^2+2\tilde m_j\sum_{i=j+1}^{N}\tilde m_ie^{{\tilde x}_j-{\tilde x}_i}\right)\left(\sum_{i=1}^{j}\tilde m_ie^{{\tilde x}_i-{\tilde x}_j}+\sum_{i=j+1}^{N}\tilde m_ie^{{\tilde x}_j-{\tilde x}_i}\right)e^{-{\tilde x}_j}\\
=&\sum_{j=1}^k\left(e^{2{\tilde x}_j}-e^{2{\tilde x}_{j-1}}\right)\left(\sum_{i=j}^N\tilde m_ie^{-{\tilde x}_i}\right)\left(\sum_{i=k}^N\tilde m_ie^{-{\tilde x}_i}\right)^2\\
&+\sum_{j=k+1}^N\left(e^{2{\tilde x}_j}-e^{2{\tilde x}_{j-1}}\right)\left(\sum_{i=j}^N\tilde m_ie^{-{\tilde x}_i}\right)^3,
\end{align*}
which can be proven to be an identity after a lengthy calculation involving multiple summation operations.

Subtracting \eqref{pf_dp_x1} from \eqref{pf_dp_m1},  we obtain 
\begin{align*}
&\frac{d}{dt}\ln\left(e^{2{\tilde x}_k}-e^{2{\tilde x}_{k-1}}\right)=\frac{2\dot {\tilde x}_ke^{2{\tilde x}_k}-2\dot {\tilde x}_{k-1}e^{2{\tilde x}_{k-1}}}{e^{2{\tilde x}_k}-e^{2{\tilde x}_{k-1}}}\\ =&4\sum_{j=1}^{k-1}\left(e^{2{\tilde x}_j}-e^{2{\tilde x}_{j-1}}\right)\left(\sum_{i=j}^N\tilde m_ie^{-{\tilde x}_i}\right)\left(\sum_{i=k}^N\tilde m_ie^{-{\tilde x}_i}\right)\\
&+2\left(e^{2{\tilde x}_k}-e^{2{\tilde x}_{k-1}}\right)\left(\sum_{i=k}^N\tilde m_ie^{-{\tilde x}_i}\right)\left(\sum_{i=k}^N\tilde m_ie^{-{\tilde x}_i}\right),
\end{align*}
which yields
\begin{align*}
&2\left(\sum_{j=1}^{k}\tilde m_je^{{\tilde x}_j-{\tilde x}_k}+\sum_{j=k+1}^{N}\tilde m_je^{{\tilde x}_k-{\tilde x}_j}\right)^2e^{2{\tilde x}_k}\\
&-2\left(\sum_{j=1}^{k-1}\tilde m_je^{{\tilde x}_j-{\tilde x}_{k-1}}+\sum_{j=k}^{N}\tilde m_je^{{\tilde x}_{k-1}-{\tilde x}_j}\right)^2e^{2{\tilde x}_{k-1}} \\ 
=&4\sum_{j=1}^{k-1}\left(e^{2{\tilde x}_j}-e^{2{\tilde x}_{j-1}}\right)\left(\sum_{i=j}^N\tilde m_ie^{-{\tilde x}_i}\right)\left(\sum_{i=k}^N\tilde m_ie^{-{\tilde x}_i}\right)\left(e^{2{\tilde x}_k}-e^{2{\tilde x}_{k-1}}\right)\nonumber\\
&+2\left(e^{2{\tilde x}_k}-e^{2{\tilde x}_{k-1}}\right)\left(\sum_{i=k}^N\tilde m_ie^{-{\tilde x}_i}\right)\left(\sum_{i=k}^N\tilde m_ie^{-{\tilde x}_i}\right)\left(e^{2{\tilde x}_k}-e^{2{\tilde x}_{k-1}}\right),
\end{align*}
when the evolution relation for $\tilde x_k$ is used. The validity of the above equality can be confirmed, again, after a lengthy calculation involving multiple summation operations. Therefore, we complete the proof for the converse part. 

Finally, the bijection is concluded by observing that the two maps \eqref{xm_to_yn} and \eqref{yn_to_xm} are inverse to each other.
\end{proof}

As an immediate corollary, we obtain a correspondence between their explicit solutions.
\begin{coro} There exists a bijection between the explicit formulae \eqref{sol:DP_FG} for N-peakon solution to the DP equation and those \eqref{NV_form_yn} for the NV equation according to the maps \eqref{xm_to_yn} and \eqref{yn_to_xm}.
\end{coro}
\begin{proof}
First of all, substituting \eqref{sol:DP_FG} into \eqref{xm_to_yn}  gives
\begin{align*}
{\tilde x}_{N+1-k}&=\frac{1}{2}\ln\left(\sum_{j=k}^Nm_{N+1-j}e^{2x_{N+1-j}}\right)=\frac{1}{2}\ln\left(\sum_{j=k}^N\frac{2(F_{j}^{(1,0)})^2}{F_{j}^{(1,1)}F_{j-1}^{(1,1)}}\right).
\end{align*}
By employing the identity \eqref{id:bi1}, we then obtain
\begin{align*}
{\tilde x}_{N+1-k}&=\frac{1}{2}\ln\left(2\sum_{j=k}^N\frac{F_{j}^{(0,0)}}{F_{j-1}^{(1,1)}}-\frac{F_{j+1}^{(0,0)}}{F_{j}^{(1,1)}}\right)=\frac{1}{2}\ln\frac{F_{k}^{(0,0)}}{F_{k-1}^{(1,1)}}=\frac{1}{2}\ln\frac{Z_k}{W_{k-1}}.
\end{align*}
The proof for $\tilde m_k$ is much more straightforward since
\begin{align*}
 \tilde m_{N+1-k}&=\left(e^{-x_{N+1-k}}-e^{-x_{N+2-k}}\right)e^{{\tilde x}_{N+1-k}}=\left(\frac{\tau_{k-1}^{(1)}}{\tau_k^{(0)}}-\frac{\tau_{k-2}^{(1)}}{\tau_{k-1}^{(0)}}\right)\sqrt\frac{Z_k}{W_{k-1}}.
\end{align*}

Conversely, it follows that by substituting \eqref{NV_form_yn} into \eqref{yn_to_xm} 
\begin{align*}
 x_{N+1-k}&=-\ln\left(\sum_{j=1}^k\tilde m_{N+1-j}e^{-{\tilde x}_{N+1-j}}\right)=-\ln\left(\sum_{j=1}^k\frac{W_{j-1}}{U_jU_{j-1}}\right)\\
 &=-\ln\left(\sum_{j=1}^k\frac{\tau_{j-1}^{(1)}}{\tau_j^{(0)}}-\frac{\tau_{j-2}^{(1)}}{\tau_{j-1}^{(0)}}\right)=-\ln\frac{\tau_{k-1}^{(1)}}{\tau_k^{(0)}}.\\
\end{align*}
After substitution and employing the identity \eqref{id:bi1} again, we also derive
\begin{align*}
 \qquad m_k=\left(e^{2{\tilde x}_k}-e^{2{\tilde x}_{k-1}}\right)e^{-2 x_k}=\left(\frac{2F_k^{(0,0)}}{F_{k-1}^{(1,1)}}-\frac{2F_{k+1}^{(0,0)}}{F_{k}^{(1,1)}}\right)\frac{F_{k-1}^{(2,1)}}{2F_{k}^{(1,0)}}=\frac{F_{k}^{(1,0)}F_{k-1}^{(2,1)}}{F_{k}^{(1,1)}F_{k-1}^{(1,1)}},
\end{align*}
which completes the proof.
\end{proof}

\begin{remark} 
It is not difficult to verify that the long time asymptotics in \cite[Theorem 2.25]{lundmark2005degasperis} and \cite[Theorem 9.4]{hone2009explicit} for DP and Novikov peakons can also be derived from each other by using our Theorem \ref{th:dpnv}.
\end{remark}

\subsection{On constants of motions}
Recall that the DP multipeakon dynamical system admits $N$ functionally independent constants of motion given by (\cite[Theorem 2.10]{lundmark2005degasperis})
\begin{align}\label{cons:dp}
M_k=\sum_{I\in\binom{[N]}{k}}\left(\prod_{i\in I}m_i\right)\left(\prod_{j=1}^{k-1}\left(1-e^{x_{i_j}-x_{i_{j+1}}}\right)^2\right)
\end{align}
for $k=1,2,\ldots,N$, where $\binom{[N]}{k}$ denotes the set of all $k$-element subsets $I=\{i_1, i_2,\dots, i_k\}$ of $[N]=\{1,2,\ldots,N\}$ such that $i_1 < i_2 < \dots < i_k$. 
In particular, one has 
\begin{align*}
M_1=\sum_{j=1}^Nm_j,\qquad\quad M_N=\prod_{j=1}^Nm_j\prod_{j=1}^{N-1}(1-e^{x_j-x_{j+1}}).
\end{align*}

As mentioned in \cite[Theorem 4.2]{hone2009explicit}, the Novikov multipeakon dynamical system also admits $N$  constants of motion $H_1,H_2,\ldots, H_N$, where $H_k$ equals the sum of all $k\times k$ minors (principle and non-principle) of the $N\times N$ symmetric matrix ($PEP$ in \cite{hone2009explicit})
$$\left(\tilde m_i\tilde m_je^{-|\tilde x_i-\tilde x_j|}\right)_{i,j=1}^N.$$
To the best of our knowledge,  all of these quantities have never been explicitly written down, except for
\begin{align*}
&H_1=\sum_{i,j=1}^N\tilde m_i\tilde m_j e^{-|\tilde x_i-\tilde x_j|},\qquad\quad H_N=\prod_{j=1}^{N-1}\left(1-e^{-2|\tilde x_j-\tilde x_{j+1}|}\right)\prod_{j=1}^N {\tilde m_j}^2.
\end{align*}
However, $N$ constants of motion with explicit expressions for the Novikov peakon ODEs immediately follow from the known constants of motion \eqref{cons:dp} for the DP peakon ODEs and the obtained formula \eqref{yn_to_xm} in Theorem \ref{th:dpnv}. Recall that the constants of motions are the coefficients of the polynomials $\phi_1(1;z)$ and $\tilde \phi_3(1;z)$ in the discrete (dual) cubic strings \cite{lundmark2005degasperis,hone2009explicit}, and  we have the relation $\tilde \phi_3(1;z)=-z\phi_1(1;z)$ obtained by  \eqref{rel:wave}. We finally conclude with the following result.
\begin{coro} 
The Novikov multipeakon dynamical system admits $N$  constants of motion $\tilde M_1,\tilde M_2,\ldots, \tilde M_N$ explicitly given by
\begin{align*}
\tilde M_k=\sum_{I\in\binom{[N]}{k}}\left(\prod_{i\in I}(e^{2\tilde x_i}-e^{2\tilde x_{i-1}})\right)\left(\sum_{j=i_k}^N\tilde m_j e^{-\tilde x_j}\right)^2\prod_{j=1}^{k-1}\left(\sum_{l=i_j}^{i_{j+1}-1}\tilde m_l e^{-\tilde x_l}\right)^2,
\end{align*}
and each $\tilde M_k$ is equal to $H_k$ as well as the sum of all $k\times k$ minors of the $N\times N$ symmetric matrix 
$$\left(\tilde m_i\tilde m_je^{-|\tilde x_i-\tilde x_j|}\right)_{i,j=1}^N.$$
\end{coro}
\begin{remark}
It is not hard to calculate that  $\tilde M_1$ and $\tilde M_N$ are equal to $H_1$ and $H_N$, respectively. 
\end{remark}

\begin{remark}
It is argued in \cite{kardell:2016:phdthesis,kardell-lundmark} that mixed peakon-antipeakon solutions to the Novikov equation exhibit a much greater variety of possible behaviours than those of the CH equation. Additionally, the collision dynamics of the Novikov peakon-antipeakon collisions are also different from those of the DP peakon-antipeakons \cite{szmigielski-zhou:shocks-DP, szmigielski-zhou:DP-peakon-antipeakon}, which, in particular, can result in the formation of so-called shockpeakons \cite{lundmark2007formation} which are absent in the Novikov case. The mixed peakon-antipeakon problem of the Novikov equation still requires further investigation. Our results regarding explicit constants of motion could be used directly to predict some characteristics of the Novikov peakon-antipeakon collision behaviour.  
\end{remark}

\subsection{On Poisson structures} Recall that the DP peakon dynamical system \eqref{DP_eq:peakon} can be written as a non-canonical Hamiltonian system 
\begin{align}\label{ham_dp}
\dot x_k=\{x_k,M_1\}_{dp},\qquad \dot m_k=\{m_k,M_1\}_{dp}
\end{align}
with the Hamiltonian $M_1=\sum_{j=1}^Nm_j$ and the non-canonical Poisson brackets \cite[eq. (3.20)]{degasperis2003integrable}
\begin{subequations}\label{PB_dp}
\begin{align}
\{x_j,x_k\}_{dp}&=\frac{1}{2}\sgn(x_j-x_k)\left(1-e^{-|x_j-x_k|}\right),\\
\{x_j,m_k\}_{dp}&=m_ke^{-|x_j-x_k|},\\
\{m_j,m_k\}_{dp}&=2\sgn(x_j-x_k)m_jm_ke^{-|x_j-x_k|}.
\end{align}
\end{subequations}
On the other hand,
the Novikov peakon dynamical system \eqref{NV_eq:peakon} can be
equipped with the non-canonical Poisson structure \cite[eq. (3.5)]{hone2009explicit}
\begin{align}\label{ham_nv}
\dot {\tilde x}_k=\{\tilde x_k,\tilde M_1\}_{nv},\qquad \dot {\tilde m}_k=\{\tilde m_k,\tilde M_1\}_{nv},
\end{align}
where  the Hamiltonian is $\tilde M_1=\sum_{i,j=1}^N\tilde m_i\tilde m_j e^{-|\tilde x_i-\tilde x_j|}$ and
\begin{subequations}\label{PB_nv}
\begin{align}
\{\tilde x_j,\tilde x_k\}_{nv}&=\frac{1}{2}\sgn(\tilde x_j-\tilde x_k)\left(1-e^{-2|\tilde x_j-\tilde x_k|}\right),\\
\{\tilde x_j,\tilde m_k\}_{nv}&=\frac{1}{2}\tilde m_ke^{-2|\tilde x_j-\tilde x_k|},\\
\{\tilde m_j,\tilde m_k\}_{nv}&=\frac{1}{2}\sgn(\tilde x_j-\tilde x_k)\tilde m_j\tilde m_ke^{-2|\tilde x_j-\tilde x_k|}.
\end{align}
\end{subequations}
It is noted that these formulae are the same as \cite[eq. (3.5)]{hone2009explicit}, up to a scaling factor $1/2$, and the Hamiltonian here is twice as much as that one.

As argued in the previous subsection on constants of motion, the map \eqref{xm_to_yn} (or \eqref{yn_to_xm})  connects the Hamiltonians $M_1$ and $\tilde M_1$. Besides, it follows from Theorem \ref{th:dpnv} that this map  establishes a bijection between DP and Novikov pure peakon trajectories, which implies that it preserves Hamilton’s equations \eqref{ham_dp} and \eqref{ham_nv}. Therefore, it is conceivable that this map would also connect the Poisson brackets \eqref{PB_dp} and \eqref{PB_nv} well.
\begin{theorem}
The invertible map \eqref{xm_to_yn} (or its inverse \eqref{yn_to_xm}) between $\{x_k,m_k\}_{k=1}^N$ and $\{\tilde x_k,\tilde m_k\}_{k=1}^N$ induces a canonical (or Poisson) map $\phi$ from $\left(\mathcal M_1,\{\,,\,\}_{dp}\right)$ to $\left(\mathcal M_2,\{\,,\,\}_{nv}\right)$, that is,
$$
\{\phi^*f,\phi^*g\}_{dp}=\phi^*\{f,g\}_{nv},\qquad \forall\, f,g\in \mathcal M_2.
$$
\end{theorem}
\begin{proof}
It suffices to prove, under the map \eqref{xm_to_yn} (or its inverse \eqref{yn_to_xm}), there holds
\begin{subequations}
\begin{align}
\{\tilde x_j,\tilde x_k\}_{nv}=&\sum_{i=1}^N\sum_{l=1}^N\left(\frac{\partial \tilde x_j}{\partial x_i}\{x_i,x_l\}_{dp}\frac{\partial \tilde x_k}{\partial x_l}+\frac{\partial \tilde x_j}{\partial x_i}\{x_i,m_l\}_{dp}\frac{\partial \tilde x_k}{\partial m_l}\right.\nonumber\\
&\qquad\quad\left.+\frac{\partial \tilde x_j}{\partial m_i}\{m_i,x_l\}_{dp}\frac{\partial \tilde x_k}{\partial x_l}+\frac{\partial \tilde x_j}{\partial m_i}\{m_i,m_l\}_{dp}\frac{\partial \tilde x_k}{\partial m_l}\right),\\
\{\tilde x_j,\tilde m_k\}_{nv}=&\sum_{i=1}^N\sum_{l=1}^N\left(\frac{\partial \tilde x_j}{\partial x_i}\{x_i,x_l\}_{dp}\frac{\partial \tilde m_k}{\partial x_l}+\frac{\partial \tilde x_j}{\partial x_i}\{x_i,m_l\}_{dp}\frac{\partial \tilde m_k}{\partial m_l}\right.\nonumber\\
&\qquad\quad\left.+\frac{\partial \tilde x_j}{\partial m_i}\{m_i,x_l\}_{dp}\frac{\partial \tilde m_k}{\partial x_l}+\frac{\partial \tilde x_j}{\partial m_i}\{m_i,m_l\}_{dp}\frac{\partial \tilde m_k}{\partial m_l}\right),\\
\{\tilde m_j,\tilde m_k\}_{nv}=&\sum_{i=1}^N\sum_{l=1}^N\left(\frac{\partial \tilde m_j}{\partial x_i}\{x_i,x_l\}_{dp}\frac{\partial \tilde m_k}{\partial x_l}+\frac{\partial \tilde m_j}{\partial x_i}\{x_i,m_l\}_{dp}\frac{\partial \tilde m_k}{\partial m_l}\right.\nonumber\\
&\qquad\quad\left.+\frac{\partial \tilde m_j}{\partial m_i}\{m_i,x_l\}_{dp}\frac{\partial \tilde m_k}{\partial x_l}+\frac{\partial \tilde m_j}{\partial m_i}\{m_i,m_l\}_{dp}\frac{\partial \tilde m_k}{\partial m_l}\right), \label{pos_map_mm}
\end{align}
\end{subequations}
which can be equivalently written in a matrix form
\begin{align}
\Pi_{nv}=J\,\Pi_{dp}\,J^\top,
\end{align}
where $\Pi_{dp}$ and $\Pi_{nv}$ are the respective Poisson matrices $$
\Pi_{dp}=
\begin{pmatrix}
\{x_j,x_k\}_{dp}&\{x_j,m_k\}_{dp}\\
\{m_j,x_k\}_{dp}&\{m_j,m_k\}_{dp}
\end{pmatrix}_{j,k=1}^N,
\qquad 
\Pi_{nv}=
\begin{pmatrix}
\{\tilde x_j,\tilde x_k\}_{nv}&\{\tilde x_j,\tilde m_k\}_{nv}\\
\{\tilde m_j,\tilde x_k\}_{nv}&\{\tilde m_j,\tilde m_k\}_{nv}
\end{pmatrix}_{j,k=1}^N, 
$$
and $J$ is the Jacobian matrix
$$
J=\frac{\partial (\tilde x_1,\ldots,\tilde x_N; \tilde m_1,\ldots,\tilde m_N)}{\partial ( x_1,\ldots,x_N;  m_1,\ldots,m_N)}.
$$
This can actually be demonstrated after a lengthy but straightforward calculation.  

\end{proof}


\section{A new integrable lattice connecting B-Toda and C-Toda lattices}\label{sec:bc}
The B-Toda and C-Toda lattices can be regarded as positive flows obtained by isospectral deformations 
\begin{equation} \label{evo_post}
\dot \zeta_j(t)=0,\qquad \dot b_j(t)=\zeta_jb_j(t)
\end{equation}
related to the (dual) cubic strings, respectively.  In fact, they are directly derived by  isospectral deformations of PSOPs and CBOPs, which naturally arise in the transition matrices of the discrete (dual) cubic strings. In addition,  recall that, determinant tau functions play important roles in the C-Toda lattice \cite{chang2018degasperis}, while Pfaffian tau functions naturally appear in the B-Toda lattice \cite{chang2018application}, due to the fact that their respective hierarchies are associated with different types of infinite dimensional Lie algebras.   In other words, the C-Toda lattice can be transformed into bilinear equations involving certain determinant identities, while the B-Toda lattice corresponds to certain Pfaffian identities.

In this section, we establish a bridge between B-Toda and C-Toda lattices. In fact, we propose a new integrable lattice that connects both of the B-Toda and C-Toda lattices. A key ingredient is the discovery of  several bilinear identities that involve both determinants and Pfaffians.
To begin with, we note that there hold the following derivation formula (see \cite[Lemma 3.2]{chang2018degasperis}). 

\begin{lemma}  Under the condition \eqref{evo_post}, the determinants $F_{k}^{(i,j)}$ in Definition \ref{def:FG} admit the evolution
\begin{equation}\label{deriv_F}
\dot F_{k}^{(i,j)}=-E_{k+1}^{(i,j)}.
\end{equation}
\end{lemma}

\subsection{Bilinear relations connecting determinants and Pfaffians} \label{subsec:detpf}
The following bilinear relations involving only determinants are important to us.
\begin{lemma}\label{lem_new_bi}
For any $l\in \mathbb{Z}$, $k\in  \mathbb{N}_+$, there hold 
\begin{align*}
E_{k+1}^{(l,l)}F_{k-1}^{(l+1,l)}=E_{k}^{(l,l+1)}F_{k}^{(l,l)}-G_{k}^{(l,l+1)}G_{k}^{(l,l)},\\
G_{k+1}^{(l,l)}G_{k}^{(l,l+1)}=E_{k+1}^{(l,l)}F_{k}^{(l,l+1)}-E_{k+1}^{(l,l+1)}F_{k}^{(l,l)},\\
F_{k+1}^{(l,l)}G_{k}^{(l,l+1)}=G_{k+1}^{(l,l)}F_{k}^{(l,l+1)}-G_{k+1}^{(l,l+1)}F_{k}^{(l,l)}.
\end{align*}
\end{lemma}
\begin{proof}
These relations can be verified by employing the Desnanot--Jacobi identity \cite[Section 2.3, Proposition 10]{krattenthaler}, 
i.e., 
\begin{align*}
\mathcal{D} \mathcal{D}\left(\begin{array}{cc}
i_1 & i_2 \\
j_1 & j_2 \end{array}\right)=\mathcal{D}\left(\begin{array}{c}
i_1  \\
j_1 \end{array}\right)\mathcal{D}\left(\begin{array}{c}
i_2  \\
j_2 \end{array}\right)-\mathcal{D}\left(\begin{array}{c}
i_1  \\
j_2 \end{array}\right)\mathcal{D}\left(\begin{array}{c}
i_2  \\
j_1 \end{array}\right),\
\end{align*}
where $\mathcal{D}\left(\begin{array}{cccc}
i_1&i_2 &\cdots& i_k\\
j_1&j_2 &\cdots& j_k
\end{array}\right)$, for $ i_1<i_2<\cdots<i_k,\ j_1<j_2<\cdots<j_k$, denotes the
determinant of the matrix obtained from $\mathcal{D}$ by removing the rows with indices
$i_1,i_2,\dots, i_k$ and the columns with indices $j_1,j_2,\dots, j_k$.

By considering 
\begin{align*}
&\mathcal{D}_1=E_{k+1}^{(l,l)}, &&i_1=1,\,  j_1=k,\, i_2=j_2=k+1,\\
&\mathcal{D}_2=
\det\left(\begin{array}{cc}
0&1\\
I_{l+p,l+q}&\beta_{l+p}\\
\beta_{l+q}&0
\end{array}
\right)_{\substack{p=0,\ldots,k-2\\ q=0,\ldots,k-1}}, && i_1=j_1=1,\,  i_2=k+1,\, j_2=k,\\
&\mathcal{D}_3=
\det\left(\begin{array}{cc}
0&1\\
I_{l+p,l+q}&\beta_{l+p}
\end{array}
\right)_{p,q=0}^{k-1}, && i_1=j_1=1,\,  i_2=k+1,\, j_2=k,
\end{align*}
the desired bilinear relations are nothing but consequences of applying  the Desnanot--Jacobi identity to $\mathcal{D}_1,\mathcal{D}_2,\mathcal{D}_3$, respectively.
\end{proof}
By making use of the relations for $F_{k}^{(l,l+1)},G_{k}^{(l,l+1)}$ in \eqref{rel:FGtau} as well as the evolution relations \eqref{deriv_F}, it is not hard to obtain the following crucial corollary from the above lemma, which contains three bilinear relations involving both determinants and Pfaffians.  




\begin{theorem}\label{coro_new_bi}
For any $l\in \mathbb{Z}$, $k\in  \mathbb{N}_+$, there hold 
\begin{subequations}\label{eq:new_bi}
\begin{align}
\dot F_{k}^{(l,l)}\tau_{k-1}^{(l)}-2F_{k}^{(l,l)}\dot\tau_{k-1}^{(l)}=G_{k}^{(l,l)}\tau_{k}^{(l)},\\
2F_{k}^{(l,l)}\dot\tau_{k}^{(l)}-\dot F_{k}^{(l,l)}\tau_{k}^{(l)}=2G_{k+1}^{(l,l)}\tau_{k-1}^{(l)},\\
F_{k}^{(l,l)}\tau_{k+1}^{(l)}+2F_{k+1}^{(l,l)}\tau_{k-1}^{(l)}=G_{k+1}^{(l,l)}\tau_{k}^{(l)}.
\end{align}
\end{subequations}
\end{theorem}
Now we are ready to derive an integrable lattice that connects both of the B-Toda and C-Toda lattices.

\subsection{A new integrable lattice and its relations with B-Toda and C-Toda lattices}
We mainly make use of Theorem \ref{coro_new_bi}, with focus on the case of $l=0$ for the sake of convenience.
Introduce  four nonlinear variables $\alpha_k,\beta_k,\gamma_k,\delta_k$ given by
\begin{align}\label{def_ABCD}
\alpha_k=\frac{\tau_{k}^{(0)}}{\sqrt{F_{k}^{(0,0)}}},\qquad \beta_k=\frac{\tau_{k-1}^{(0)}}{\sqrt{F_{k}^{(0,0)}}},\qquad \gamma_k=\frac{G_{k+1}^{(0,0)}}{F_{k}^{(0,0)}}, \qquad \delta_k=\frac{G_{k}^{(0,0)}}{F_{k}^{(0,0)}},
\end{align}
where $\tau_k^{(0)}$ is the Pfaffian in Definition \ref{def:tau} and $F_{k}^{(0,0)}$, $G_{k}^{(0,0)}$ are determinants in Definition \ref{def:FG}.
By employing Theorem \ref{coro_new_bi}, we demonstrate that the following nonlinear relations hold.
\begin{lemma} The nonlinear variables $\alpha_k,\beta_k,\gamma_k,\delta_k$ defined in \eqref{def_ABCD} satisfy
\begin{align*}
&\dot \alpha_k=\gamma_k\beta_k,&& \dot \beta_k=-\frac{1}{2}\delta_k\alpha_k,\\
&\gamma_k=\frac{1}{\beta_{k+1}^2}(\alpha_{k+1}\beta_{k+1}+2\alpha_k\beta_k), && \delta_k=\frac{1}{\alpha_{k-1}^2}(\alpha_k\beta_k+2\alpha_{k-1}\beta_{k-1}).
\end{align*}
\end{lemma}
\begin{proof}
The evolution relation for $\alpha_k$ can be demonstrated by the second equality in \eqref{eq:new_bi}, while that for $\beta_k$ can be confirmed by the first equality in \eqref{eq:new_bi}. The remaining two nonlinear relations without derivative follow from the third equality in \eqref{eq:new_bi}.
\end{proof}
By eliminating $\gamma_k$ and $\delta_k$, we are immediately led to a couple of lattices containing only $\alpha_k$ and $\beta_k$.
\begin{theorem} The nonlinear variables $\alpha_k,\beta_k$ defined in \eqref{def_ABCD} satisfy
\begin{subequations}\label{eq:new_IS}
\begin{align}
\dot \alpha_k&=\frac{\beta_k}{\beta_{k+1}^2}(\alpha_{k+1}\beta_{k+1}+2\alpha_k\beta_k),\\
\dot \beta_k&=-\frac{\alpha_k}{2\alpha_{k-1}^2}(\alpha_k\beta_k+2\alpha_{k-1}\beta_{k-1}).
\end{align}
\end{subequations}
with $k=1,2,\ldots$ and $\alpha_0=1,\beta_0=0.$
\end{theorem}

The argument above implies  that we derive a lattice equation \eqref{eq:new_IS} that admits solutions in terms of \eqref{def_ABCD}. In fact, this lattice equation is also integrable since we can relate it with both of the B-Toda and C-Toda lattices.

Recall that the B-Toda lattice (see e.g. \cite{chang2018partial,chang2018application,chang2022hermite}) reads
\begin{align}\label{eq:btoda}
&\dot {\tilde v}_{k}=\tilde v_{k}(\tilde d_{k}-\tilde d_{k-1}),  \qquad \dot {\tilde d}_k=\tilde v_{k+1}(\tilde d_{k+1}+\tilde d_k)-\tilde v_{k}(\tilde d_{k}+\tilde d_{k-1})
\end{align}
with $k=0,1,2,\ldots$  and $\tilde v_0=0.$
It can be obtained from the compatibility of the four-term recurrence relation 
$$
z(P_k^B-\tilde v_kP_{k-1}^B)=P_{k+1}^B+(\tilde d_{k}-\tilde v_k)P_k^B+\tilde v_k(\tilde d_k-\tilde v_{k+1})P_{k-1}^B+(\tilde v_k)^2\tilde v_{k-1}P_{k-2}^B,
$$
and the evolution relation
$$\dot P_k^B-\tilde v_k\dot P_{k-1}^B=-\tilde v_k(\tilde d_k+\tilde d_{k-1})P_{k-1}^B,$$
for PSOPs  $\{P_k^B(z;t)\}$.
Consequently, the B-Toda lattice can be equivalently written as
$$\dot {\tilde L}=[\tilde B,\tilde L],$$
where 
\begin{eqnarray*}
&\tilde L=\tilde L_1^{-1}\tilde L_2,\qquad \tilde B=\tilde L_1^{-1}\tilde B_2,\\
&\tilde L_1=\left(\begin{array}{cccccc}
1&&\\
-\tilde v_1&1&\\
&-\tilde v_{2}&1\\
&&\ddots&\ddots\\
\end{array}
\right),\quad
\tilde B_2=\left(\begin{array}{cccccc}
0&&&\\
\tilde{\mathcal{A}}_1&0&&\\
&\tilde{\mathcal{A}}_{2}&0\\
&&\ddots&\ddots
\end{array}
\right),
\end{eqnarray*}
\begin{eqnarray*}
\tilde L_2=\left(\begin{array}{cccccc}
\tilde{\mathcal{B}}_0&1&&&\\
\tilde{\mathcal{C}}_1&\tilde{\mathcal{B}}_1&1&&\\
\tilde{\mathcal{D}}_2&\tilde{\mathcal{C}}_2&\tilde{\mathcal{B}}_2&1\\
&\tilde{\mathcal{D}}_{3}&\tilde{\mathcal{C}}_{3}&\tilde{\mathcal{B}}_{3}&1\\
&&\ddots&\ddots&\ddots&\ddots\\
\end{array}
\right)
\end{eqnarray*}
with $$\tilde{\mathcal{A}}_k=-\tilde v_k(\tilde d_k+\tilde d_{k-1}),\quad \tilde{\mathcal{B}}_k=-\tilde v_k+\tilde d_{k},\quad \tilde{\mathcal{C}}_k=\tilde v_k(\tilde d_k-\tilde v_{k+1}),\quad \tilde{\mathcal{D}}_k=(\tilde v_k)^2\tilde v_{k-1}.$$
Moreover, it is noted that the B-Toda lattice admits solutions 
\begin{subequations}\label{b-sol}
\begin{align}
\tilde v_k&=\frac{\tau_{k+1}^{(0)}\tau_{k-1}^{(0)}}{({\tau_{k}^{(0)}})^2},\\
\tilde d_k&=\left(\ln\frac{\tau_{k+1}^{(0)}}{\tau_{k}^{(0)}}\right)_t=
\frac{\tau_{k+1}^{(0)}\tau_{k-1}^{(0)}}{({\tau_{k}^{(0)}})^2}+\frac{\tau_{k+2}^{(0)}\tau_{k}^{(0)}}{({\tau_{k+1}^{(0)}})^2}+\frac{(\tau_{k+1}^{(0)})^2F_{k}^{(0,0)}}{2(\tau_{k}^{(0)})^2F_{k+1}^{(0,0)}}+\frac{2(\tau_{k}^{(0)})^2F_{k+2}^{(0,0)}}{(\tau_{k+1}^{(0)})^2F_{k+1}^{(0,0)}}.
\end{align}
\end{subequations}

On the other hand, recall that the C-Toda lattice (see e.g. \cite{chang2018degasperis})
\begin{align}\label{eq:ctoda}
&\dot v_{k}=v_{k}(d_{k}-d_{k-1}),  \qquad \dot d_k=2(\sqrt{ v_{k+1}d_{k+1}d_{k} }-\sqrt{ v_kd_kd_{k-1} })
\end{align}
with $k=0,1,2,\ldots$  and $v_0=0$, 
can be obtained from the compatibility of the four-term recurrence relation 
\begin{align*}
xL_1P^C(z;t)=L_2P^C(z;t)
\end{align*}
and the evolution relation
\begin{align*}
L_1\dot P^C(z;t) =B_2P^C(z;t),
\end{align*}
for the symmetric CBOPs  $\{P_k^C(z;t)\}$, that is, the C-Toda lattice admits the Lax representation:
$$\dot L=[B,L],$$
where 
$$L=L_1^{-1}L_2,\qquad B=L_1^{-1}B_2$$
\begin{eqnarray*}
&L_1=\left(\begin{array}{cccccc}
1&&\\
\mathcal{A}_1&1&\\
&\mathcal{A}_{2}&1\\
&&\ddots&\ddots
\end{array}
\right),\quad L_2=\left(\begin{array}{cccccc}
\mathcal{B}_0&1&&&\\
\mathcal{C}_1&\mathcal{B}_1&1&&\\
\mathcal{D}_2&\mathcal{C}_2&\mathcal{B}_2&1\\
&\mathcal{D}_{3}&\mathcal{C}_{3}&\mathcal{B}_{3}&1\\
&&\ddots&\ddots&\ddots&\ddots
\end{array}
\right),
\end{eqnarray*}
\begin{eqnarray*}
B_2=\left(\begin{array}{cccccc}
0&&&\\
2(\mathcal{B}_0-\mathcal{A}_0)\mathcal{A}_1&0&&\\
&2(\mathcal{B}_{1}-\mathcal{A}_{1})\mathcal{A}_{2}&0\\
&&\ddots&\ddots
\end{array}
\right),
\end{eqnarray*}
with $$\mathcal{A}_k=-\sqrt{\frac{d_kv_k}{d_{k-1}}},\quad \mathcal{B}_k=\frac{1}{2}d_k-\sqrt{\frac{d_kv_k}{d_{k-1}}},\quad \mathcal{C}_k=-v_k+\frac{1}{2}\sqrt{d_kv_kd_{k-1}},\quad \mathcal{D}_k=v_{k-1}\sqrt{\frac{d_kv_k}{d_{k-1}}}.$$
In addition, it is noted that the C-Toda lattice enjoys the solution
\begin{equation}\label{c-sol}
 v_k=\frac{F_{k+1}^{(0,0)}F_{k-1}^{(0,0)}}{({F_{k}^{(0,0)}})^2},\qquad  d_k=\left(\ln\frac{F_{k+1}^{(0,0)}}{F_{k}^{(0,0)}}\right)_t=\frac{(G_{k+1}^{(0,0)})^2}{F_{k+1}^{(0,0)}F_{k}^{(0,0)}}.
\end{equation}

Motivated by the expressions of solutions \eqref{def_ABCD} , \eqref{b-sol} and \eqref{c-sol}, we now relate the  lattice equation \eqref{eq:new_IS} with both of the B-Toda and C-Toda lattices in the following ways. 
\begin{theorem}\label{th:BT}
The lattice equation \eqref{eq:new_IS} can be transformed into the B-Toda lattice \eqref{eq:btoda} according to changes of variables
\begin{equation*}
\tilde v_k=\frac{\alpha_{k+1}\beta_k}{\alpha_{k}\beta_{k+1}},\qquad \tilde d_k=\frac{\alpha_{k+1}\beta_k}{\alpha_{k}\beta_{k+1}}+\frac{\alpha_{k+2}\beta_{k+1}}{\alpha_{k+1}\beta_{k+2}}+\frac{(\alpha_{k+1})^2}{2(\alpha_{k})^2}+\frac{2(\beta_{k+1})^2}{(\beta_{k+2})^2},
\end{equation*}
and  into the C-Toda lattice \eqref{eq:ctoda} according to
\begin{equation*}
 v_k=\frac{\alpha_{k}^2\beta_{k}^2}{\alpha_{k-1}^2\beta_{k+1}^2},\qquad  d_k= \left(\frac{\alpha_{k+1}}{\alpha_{k}}+2\frac{\beta_k}{\beta_{k+1}}\right)^2.
\end{equation*}
\end{theorem}
\begin{proof}
The assertion follows by straightforward but tedious calculations. We omit the details.
\end{proof}
\begin{remark}
Theorem \ref{th:BT} provides a B\"acklund transformation between the lattice equation \eqref{eq:new_IS} and the B-Toda lattice \eqref{eq:btoda}, as well as  that between the lattice equation \eqref{eq:new_IS} and  the C-Toda lattice \eqref{eq:ctoda}. It seems that neither of the transformations is invertible, in other words, it seems not possible to write down the explicit expressions in terms of $\alpha_k=f(v_k,d_k),\beta_k=g(v_k,d_k)$  or $\alpha_k=\tilde f(\tilde v_k,\tilde d_k),\beta_k=\tilde g(\tilde v_k,\tilde d_k)$. Therefore,  the B-Toda and C-Toda lattices are not directly connected with each other in this manner, which coincides with our long standpoint.
\end{remark}
\begin{remark}
We mainly present the corresponding results for semi-infinite lattices. The finite lattices can be obtained through truncation by $v_N=0$ or $\tilde v_N=0$, or $\frac{1}{\beta_{N+1}}=0$, which correspond to discrete measures with $N$ masses.
\end{remark}

\subsection{On bilinear forms} 
The bilinear form is of importance in the theory of integrable systems (see e.g. \cite{hirota2004direct,jimbo1983solitons}). The argument in the previous subsection implies that, after changes of variables \eqref{def_ABCD}, the nonlinear lattice \eqref{eq:new_IS} indeed admits the bilinear form
\begin{subequations}\label{eq:new_bi_0}
\begin{align}
&\dot F_{k}\tau_{k-1}-2F_{k}\dot\tau_{k-1}=G_{k}\tau_{k},\label{eq:new_bi1}\\
&2F_{k-1}\dot\tau_{k-1}-\dot F_{k-1}\tau_{k-1}=2G_{k}\tau_{k-2},\label{eq:new_bi2}\\
&F_{k-1}\tau_{k}+2F_{k}\tau_{k-2}=G_{k}\tau_{k-1},\label{eq:new_bi3}
\end{align}
\end{subequations}
where the dependence of the upper index is suppressed for simplicity. In this subsection, we argue that the bilinear form \eqref{eq:new_bi_0} is more fundamental than those for B-Toda and C-Toda lattices. In fact, the conclusion immediately follows from the following result.
\begin{theorem}\label{th:bi_coro}
If the variables $F_{k},G_{k},\tau_{k}$ satisfy \eqref{eq:new_bi_0} with the usual convention, then they also satisfy the following relations:
\begin{subequations}
\begin{align}
&\dot F_{k} F_{k-1}-F_{k} \dot F_{k-1}=G_{k}^2,\label{eq:new_bi4} \\
&\dot F_{k-1} \tau_{k}+2\dot F_{k}\tau_{k-2}=2\dot \tau_{k-1}G_k,\label{eq:new_bi5} \\
&2\dot F_{k}\dot \tau_{k}-\ddot F_k \tau_{k}=\dot \tau_{k-1}G_{k+1},\label{eq:new_bi6}\\
&\ddot F_{k} F_{k}-(\dot F_{k})^2=2G_{k+1}G_{k},\label{eq:new_bi7}\\
&\ddot \tau_k \tau_k-(\dot\tau_k)^2=\tau_{k-1}\dot\tau_{k+1}-\dot\tau_{k-1}\tau_{k+1}.\label{eq:new_bi8}
\end{align}
\end{subequations}
\end{theorem}
\begin{proof}
We proceed the proof by the following order.
\begin{itemize}
\item \eqref{eq:new_bi1}, \eqref{eq:new_bi2}, \eqref{eq:new_bi3} $\Rightarrow$ \eqref{eq:new_bi4}.
By calculating $\eqref{eq:new_bi1}\times F_{k-1}+\eqref{eq:new_bi2}\times F_{k}$ and using \eqref{eq:new_bi3}, one can obtain \eqref{eq:new_bi4}.
\item \eqref{eq:new_bi1}, \eqref{eq:new_bi3},  \eqref{eq:new_bi4}$\Rightarrow$\eqref{eq:new_bi5}. \eqref{eq:new_bi5} can be derived by calculating $\eqref{eq:new_bi1}\times G_{k}$ and employing \eqref{eq:new_bi3} and  \eqref{eq:new_bi4}.
\item  \eqref{eq:new_bi1}, \eqref{eq:new_bi2}, \eqref{eq:new_bi3}, \eqref{eq:new_bi5} $\Rightarrow$ \eqref{eq:new_bi6}. Here we only sketch the derivation. First, take the derivatives of \eqref{eq:new_bi1} and \eqref{eq:new_bi2} with respect to $t$, and cancel the term $\ddot F_{k}$. Then utilize \eqref{eq:new_bi3} and \eqref{eq:new_bi5} to simplify the obtained equality.  After lengthy simplification, one can finally obtain
$$2\dot F_{k}\dot \tau_{k}-\ddot F_k \tau_{k}-\dot \tau_{k-1}G_{k+1}=cF_{k} \tau_{k},$$
where $c$ is some constant. By recalling the usual convention for $k=0$, \eqref{eq:new_bi6} follows.
\item  \eqref{eq:new_bi1}, \eqref{eq:new_bi2},  \eqref{eq:new_bi6} $\Rightarrow$ \eqref{eq:new_bi7}. By calculating $\eqref{eq:new_bi6}\times F_{k}$, and utilizing \eqref{eq:new_bi1} and \eqref{eq:new_bi2}, one could get \eqref{eq:new_bi7}.
\item  \eqref{eq:new_bi1}, \eqref{eq:new_bi2},  \eqref{eq:new_bi3}, \eqref{eq:new_bi4}, \eqref{eq:new_bi5}, \eqref{eq:new_bi6},\eqref{eq:new_bi7} $\Rightarrow$ \eqref{eq:new_bi8}. Starting from calculating $\eqref{eq:new_bi6}\times F_k$ and then making use of \eqref{eq:new_bi1}, \eqref{eq:new_bi2},  \eqref{eq:new_bi3}, \eqref{eq:new_bi4}, \eqref{eq:new_bi5}, \eqref{eq:new_bi7}, one will finally derive \eqref{eq:new_bi8}.
\end{itemize}
\end{proof}

Recall that the bilinear form of the B-Toda lattice \cite{chang2018partial,chang2018application} is given by 
\begin{equation*}
\ddot \tau_k \tau_k-(\dot\tau_k)^2=\tau_{k-1}\dot\tau_{k+1}-\dot\tau_{k-1}\tau_{k+1},
\end{equation*}
while the bilinear form of the C-Toda lattice \cite{chang2018degasperis} reads
\begin{align*}
&\dot F_{k} F_{k-1}-F_{k} \dot F_{k-1}=G_{k}^2,\\
&\ddot F_{k} F_{k}-(\dot F_{k})^2=2G_{k+1}G_{k}.
\end{align*}
It is noted that these bilinear equations all appear in Theorem \ref{th:bi_coro}.
Therefore, we immediately conclude from Theorem \ref{th:bi_coro} that the bilinear form \eqref{eq:new_bi_0} is more fundamental than those for B-Toda and C-Toda lattices.

\section*{Acknowledgements}
The author is deeply grateful to Prof. Xing-Biao Hu and Jacek Szmigielski for their  guidance  and encouragement over many years. This project was partially inspired by informal conversations with them. The author would like to dedicate this paper to Prof.  Jacek Szmigielski on the occasion of his 70th birthday.
This work was supported in part by the National Natural Science Foundation of China (\#12222119, 12288201, 12571270) and the Youth Innovation Promotion Association CAS. 

\small

\def\cydot{\leavevmode\raise.4ex\hbox{.}}
  \def\cydot{\leavevmode\raise.4ex\hbox{.}} \def\cprime{$'$}


\end{document}